\documentclass[lettersize,journal]{IEEEtran}
\usepackage{amsmath,amssymb,amsfonts}
\usepackage{algorithmic}
\usepackage{algorithm}
\usepackage{array}
\usepackage[caption=false,font=normalsize,labelfont=sf,textfont=sf]{subfig}
\usepackage{textcomp}
\usepackage{stfloats}
\usepackage{url}
\usepackage{verbatim}
\usepackage{graphicx}
\usepackage{cite}
\usepackage{hyperref}
\usepackage{pbalance}
\hypersetup{
    colorlinks = true,
    citecolor = {blue},
    linkcolor=blue
}
\usepackage{cases}
\usepackage{amsthm}
\usepackage{diagbox}
\allowdisplaybreaks

\newtheorem{theorem}{Theorem}
\newtheorem{lemma}{Lemma}

\newtheorem{corollary}{Corollary}
\newtheorem{remark}{Remark}
\newtheorem*{assumption*}{\assumptionnumber}
\providecommand{\assumptionnumber}{}
\makeatletter
\newenvironment{assumption}[1]
{%
	\renewcommand{\assumptionnumber}{A#1}%
	\begin{assumption*}%
		\protected@edef\@currentlabel{A#1}%
	}
	{%
	\end{assumption*}
}
\makeatother

\def\hyph{-\penalty0\hskip0pt\relax}

\newcommand{\tr}{^\mathrm{T}}
\newcommand{\tri}{^{-\mathrm{T}}}
\newcommand{\ntr}{^\mathrm{-T}}

\newcommand{\abs}[1]{\left\lvert#1\right\rvert}
\newcommand{\norm}[1]{\left\lVert#1\right\rVert}
\newcommand{\indnorm}[1]{{\left\vert\kern-0.25ex\left\vert\kern-0.25ex\left\vert #1 
		\right\vert\kern-0.25ex\right\vert\kern-0.25ex\right\vert}}
\newcommand{\defeq}{\doteq}
\newcommand{\confreg}[1][2]{\mathcal{C}_{#1}}
\newcommand{\outappr}[1][2]{\mathcal{O}_{#1}}

\newcommand{\BR}{\mathbb{R}}
\newcommand{\BE}{\mathbb{E}}
\newcommand{\BP}{\mathbb{P}}

\newcommand{\CB}{\mathcal{B}}

\makeatletter
\newcommand\fs@spaceruled{\def\@fs@cfont{\bfseries}\let\@fs@capt\floatc@ruled
  \def\@fs@pre{\vspace{0.5\baselineskip}\hrule height.8pt depth0pt \kern3pt}%
  \def\@fs@post{\kern3pt\hrule\vspace{-2mm}\relax}%
  \def\@fs@mid{\kern3pt\hrule\kern2pt}%
  \let\@fs@iftopcapt\iftrue}
\makeatother

\def \newtext{}

\begin{document}

\title{Finite Sample Analysis of Distribution\hyph Free Confidence Ellipsoids for Linear Regression
}

\author{Szabolcs Szentpéteri,\quad Balázs Csan\'ad Csáji
\thanks{This research was supported by the European Union within the framework of the National Laboratory for Autonomous Systems, RRF-2.3.1-21-2022-00002; and by the TKP2021-NKTA-01 grant of the National Research, Development and Innovation Office (NRDIO), Hungary.}
\thanks{Szabolcs Szentpéteri is with the Institute for Computer Science and Control (SZTAKI), Hungarian Research            Network (HUN-REN), 13-17 Kende utca, H-1111, Budapest, Hungary (e-mail: szentpeteri.szabolcs@sztaki.hu). }
\thanks{Balázs Csan\'ad Csáji is with the Institute for Computer Science and Control (SZTAKI),
		Hungarian Research Network (HUN-REN), 13-17 Kende utca, H-1111, Budapest, Hungary, and with Department of Probability Theory and Statistics, Institute of Mathematics, E\"otv\"os Lor\'and University (ELTE),
		1/C Pázmány Péter sétány, H-1117, Budapest, Hungary (e-mail: csaji@sztaki.hu).}}

\markboth{}%
{}

\maketitle

\begin{abstract}
The least squares (LS) estimate is the archetypical solution of linear regression problems. The asymptotic Gaussianity of the scaled LS error is often used to construct approximate confidence ellipsoids around the LS estimate, however, for finite samples these ellipsoids do not come with strict guarantees, unless some strong assumptions are made on the noise distributions. The paper studies the distribution-free Sign-Perturbed Sums (SPS) ellipsoidal outer approximation (EOA) algorithm which can construct non-asymptotically guaranteed confidence ellipsoids under mild assumptions, such as independent and symmetric noise terms. These ellipsoids have the same center and orientation as the classical asymptotic ellipsoids, only their radii are different, which radii can be computed by convex optimization. Here, we establish high probability non-asymptotic upper bounds for the sizes of SPS outer ellipsoids for linear regression problems and show that the volumes of these ellipsoids decrease at the optimal rate. Finally, the difference between our theoretical bounds and the empirical sizes of the regions are investigated experimentally.
\end{abstract}


\section{Introduction}
{\newtext \IEEEPARstart{B}{uilding}
mathematical models of unknown systems based on noisy empirical data is a fundamental problem in system identification, signal processing, machine learning and statistics. 
Linear stochastic systems form an important class of models \cite{kailath2000linear, caines2018}, not only because they are better-understood theoretically and can be handled more easily, but they are also widely used in practical applications. Moreover, they can serve as a stepping stone to understanding more complex models, such as nonlinear systems, which could be handled by linearization around a working point \cite{ljung2010}. 

Linear regression is the most fundamental formulation of linear estimation problems. There are several standard methods that provide point estimates for this problem with uncertainty quantification based on limiting distributions.
Typical examples include the {\em least-squares} (LS) estimate with asymptotic 
confidence regions derived from the central limit theorem \cite{Ljung1999}. A considerable disadvantage of these confidence sets is that they are only approximately correct for finite samples. Consequently, their application could be problematic in domains where {\em robustness} is an important factor or strong stability requirements are posed. Furthermore, even in not safety critical applications, the decision making process based on the estimated model could be drastically improved by taking non-asymptotically guaranteed uncertainties into account.

Robust statistical methods in linear estimation are crucial for signal processing applications \cite{zoubir2018robust}. Due to their connections to sparse signal recovery and compressed sensing, robust linear regression methods \cite{Papageorgiou2015, Mitra2013, Liu2018} are actively studied by the signal processing community. 
These methods offer estimation techniques that are more robust than LS, allowing for more effective outlier management.
We would argue that probabilistic uncertainty quantification for the LS estimate in this case could also be a possible direction. Robust estimation and uncertainty quantification is important in other applications, as well, such as direction-of-arrival estimation \cite{Dai2018, Yin2011}, channel estimation \cite{Dai2020} and object tracking \cite{Feldmann2011}. In these works, Bayesian approaches and covariance estimation are used to quantify the uncertainty or to build confidence regions. Non-Bayesian confidence region constructions were also applied for source localization in wireless sensor networks \cite{han2018}, while \cite{Ghaderpour2024} applied least squares wavelet spectrogram (LSWS) to analyze global navigation satellite systems (GNSS) time series data and built confidence level surfaces to quantify the uncertainty.

To overcome the issue of only asymptotically guaranteed confidence regions, a substantial amount of recent research in system identification focus on approaches with {\em non-asymptotic} guarantees \cite{Algo2018}.} One of the most common techniques is to derive {\em Probably Approximately Correct} (PAC) bounds for an estimate, such as LS, using the theory of concentration inequalities, which then induce confidence regions.

{\newtext Several works analyzed (closed-loop) dynamical systems, where the sample complexity of transition matrix estimations were studied for fully observable \cite{simchowitz18a} and also non-observable state space settings \cite{zheng2020non, oymak2021}. Even unstable linear dynamical systems were studied in \cite{faradonbeh2018}.} A signal processing focused approach was investigated in \cite{basu2019lowrank}, where a low-rank and structured sparse high-dimensional vector autoregressive (VAR) system was analyzed, and a new estimation technique was proposed with non-asymptotic probabilistic upper bounds on the estimation error.
In the aforementioned works, there are strong 
assumptions on the noises, namely that they follow some specific distribution, in most cases they are {\em Gaussian}. {\newtext Even the least restrictive results assume subgaussianity \cite{sarkar2021, Jedra2023}, which covers a wide range of possible noises where the tail probabilities decay at least as fast as for the Gaussian distribution.} Another limitation is that the confidence regions heavily rely on specific constants (hyper-parameters), which are coming from the assumptions on the system setting and the distributions of the inputs and the outputs. In practice, these hyper-parameters are usually unknown and only estimates can be used or bounds based on 
domain specific knowledge of experts. For finite impulse response (FIR) systems, PAC bounds for the LS estimate were given in \cite{Djehiche2021}, under a centered {\em subgaussian} noise assumption. Despite the
relaxed noise assumption compared to the Gaussian case, building confidence regions using this results would still rely on the variance proxy of the subgaussian distribution.

It is also possible to work with alternative loss functions, such as the Chebyshev loss (i.e., the $\ell_{\infty}$ case). Assuming uniformly distributed (therefore bounded) noises, PAC bounds can be derived for the Chebyshev estimate, as well \cite{Yi2024}.

{\em Set membership} identification is another class of methods that can provide quality tags for point estimates in the form of bounding ellipsoids \cite{milanese2013bounding}. These approaches usually also consider (closed-loop) dynamical systems, however, some of their variants address FIR models assuming that only quantized \cite{CASINI2012} of even binary \cite{Casini2011} measurements are available. In most cases of set membership identification, {\em bounded} noises are assumed, where the bound is typically assumed to be known {\em a priori}, though some recent works estimate it from the data \cite{lauricella2020set}.

As both set membership identification and statistical learning theory \cite{Tsiamis2023} based confidence region constructions require strong 
assumptions and a priori knowledge about the noises affecting the systems, algorithms that can build confidence regions in a {\em distribution-free} and
{\em data-driven} (hyper-parameter-free) way for any {\em finite sample} size are highly desirable.

Two identification algorithms that satisfy these properties (i.e., non-asymptotic, distribution-free and data-driven) are the LSCR: {\em Leave-out Sign-dominant Correlation Regions} \cite{Campi2005} and the SPS: {\em Sign-Perturbed Sums} \cite{Csaji2015} methods. They also have several distributed signal processing applications, including source localization in wireless sensor networks \cite{han2018}, and distributed evaluation of confidence sets
\cite{Zambianchi2018}. In \cite{han2018} LSCR confidence regions were built to characterize the search space of the source parameters, while \cite{Zambianchi2018} proposed an SPS-based diffusion algorithm to avoid flooding the network.

In this paper, we further study the SPS method, which can construct exact, non-asymptotic confidence regions for 
the true system parameters under the assumption that the noises are independent and symmetric about zero. SPS was originally introduced for general linear systems \cite{Csaji2012b}, but apart from its exact coverage probability, many of its theoretical guarantees are for linear regression models with independent regressors \cite{Csaji2015}. By using instrumental variables, several properties of SPS can be generalized to closed-loop state space models 
\cite{szentpeteri2023}.

{\newtext 
The fundamental construction of SPS is the indicator function, which
is a hypothesis test that checks whether a given parameter is included in the confidence region. Using the indicator function to determine all the points of the confidence set would be computationally demanding, therefore in \cite{Csaji2015} an {\em ellipsoidal outer approximation }(EOA) algorithm was proposed, which gives a compact representation of the SPS region.
An important feature of the constructed ellipsoid is that its center and its shape matrix coincide with that of the confidence ellipsoid constructed using the classical asymptotic theory, only its radius is different. This radius, which ensures finite sample guarantees, can be computed in polynomial time, by solving semidefinite optimization problems. The EOA construction of SPS was later generalized to ARX \cite{volpe2015sign} and closed-loop state space models \cite{szentpeteri2023}, as well.

Theoretical guarantees besides the exact coverage were rigorously proven for the acceptance region of the SPS indicator function, such as uniform strong consistency \cite{Weyer2017}, which means that the SPS regions almost surely shrink around the true parameters, as the sample size increases. However, this result is asymptotic, and the finite sample performance of SPS remained an open question until recently. In our previous work \cite{szentpeteri2025}, we have investigated the {\em sample complexity} of SPS for linear regression, i.e., we have derived non-asymptotic PAC bounds for the volumes of the acceptance regions of the SPS indicator function (hypothesis test). We have also shown that the sizes of the SPS indicator regions decrease at an optimal rate $\mathcal{O}(1/\sqrt{n})$, where $n$ is the sample size.}

{\newtext In this work we provide a finite sample analysis of the EOA of SPS for general linear regression problems with exogenous regressors, e.g., for FIR systems. 
The main contributions are:
\begin{enumerate}
    \item We prove high probability upper bounds on the sizes of SPS EOA confidence ellipsoids. Our analysis builds on some results from the sample complexity analysis of standard SPS regions \cite{szentpeteri2025}, however, we analyze the finite sample properties of a convex semidefinite program, which requires a significantly different approach than deriving concentration inequalities for the indicator function based construction of standard SPS.

    \item We also show that the obtained bounds are ``close'' to the bounds of the original SPS region and that they also decrease at the optimal rate, that is, $\mathcal{O}(1/\sqrt{n})$.

    \item Extensive simulation experiments are presented, as well. We demonstrate the difference between our theoretical bounds and the empirical sizes of the regions, furthermore, we show comparisons with the PAC bounds of \cite{Djehiche2021} and set membership ellipsoids.
\end{enumerate}
}
We  emphasize that although our analysis use distribution and regressor dependent parameters, the SPS EOA does not require these parameters, and it can be applied under very general assumptions,
unlike standard non-asymptotic solutions, e.g., based on PAC bounds or set membership ellipsoids.

{\newtext {\em Notations:} Throughout the paper $n$ stands for the sample size (i.e., the number of observed input-output pairs). For the sake of convenience, we denote the set $\{1,\dots,n\}$ by $[n]$. Given a real matrix $A$, $A\tr$ denotes the transpose, $\text{det}(A)$ the determinant, $A^{-1}$ the inverse, $A^\dagger$ the pseudoinverse of $A$ and $\bar{A}$ denotes $\frac{1}{n}A$.  The eigenvalues of matrix $A$ are denoted by $\lambda_i(A)$, and for symmetric matrices, we use $\lambda_{\text{min}}(\cdot)$ and $\lambda_{\text{max}}(\cdot)$ for the minimum and maximum eigenvalues, respectively. For a symmetric matrix $A$, $A \succ 0$ ($A \succeq 0$) means that it is positive (semi)definite. The Euclidean norm is denoted by $\|\cdot\|$. The notation $\mathbb{I}$ and $I$ stands for the indicator function and the identity matrix, respectively. The natural logarithm is denoted by $\log(\cdot)$. Throughout the paper $\Phi_n$ denotes the regressor matrix and its thin QR-decomposition is $\Phi_n = \Phi_{{\scriptscriptstyle Q},n}\Phi_{{\scriptscriptstyle R},n}$.}

\section{Problem setting}\label{sec:problem_setting}
This section introduces the addressed linear regression problem and presents our main assumptions. Note that the same assumptions are used as for the sample complexity analysis of the SPS-Indicator function (Algorithm \ref{alg:sps_indicator}) \cite{szentpeteri2025}.
\subsection{Data Generation}
Consider the following linear regression problem
\begin{equation}\label{equ:system}
    Y_t \,\defeq\, \varphi_t\tr \theta^* + W_t,
\end{equation}
for $t \in [n]\defeq\{1, \dots, n\}$, where $Y_t$ is the scalar output, $\varphi_t$ is a $d$-dimensional deterministic regressor, $\theta^*$ is the $d$-dimensional (constant) true parameter and $W_t$ is the (random) scalar noise. We are given a sample of size $n$ which consists of $\varphi_1, \dots, \varphi_n$ (input vectors) and $Y_1, \dots, Y_n$ (noisy, scalar outputs).

The following notation will be used throughout the paper:
\begin{align}
    &\Phi_n \defeq \begin{bmatrix}
        \varphi_1\tr \\
        \varphi_2\tr \\
        \vdots\\
        \varphi_n\tr 
    \end{bmatrix},\qquad
    w_n \defeq \begin{bmatrix}
        W_1\\
        W_2\\
        \vdots\\
        W_n
    \end{bmatrix},\qquad
    y_n \defeq \begin{bmatrix}
        Y_1\\
        Y_2\\
        \vdots\\
        Y_n
    \end{bmatrix},\\
    &	R_n \defeq \sum_{t=1}^{n}\varphi_t\varphi_t\tr  = \Phi_n\tr  \Phi_n, \quad \bar{R}_n \defeq \frac{1}{n}R_n.
\end{align}
In our analysis we focus on {\em deterministic} regressor sequences $\{\varphi_t\}$, however, our results can be straightforwardly generalized to random exogenous regressors, i.e., regressors that are independent of the noise sequence $\{W_t\}$. In order to do that the assumptions on the regressors \ref{assu:R_nonsigular}-\ref{assu:Phi_Q_coherence} must be satisfied almost surely and the derivations can be traced back to the presented one by fixing a realization of the regressors.
\subsection{Assumptions}
Our core assumptions are the following:
\begin{assumption}{1}\label{assu:noise}
The noise sequence $\{W_t\}$ is independent and contains nonatomic, $\sigma$-subgaussian random variables whose probability distributions are symmetric about zero. 
\end{assumption}
\vspace{-1mm}

An (integrable) random variable $W$ is $\sigma${\em-subgaussian}\cite[Definition 2.2]{wainwright_2019}, if 
for all $\lambda \in \mathbb{R},$ there is a 
$\sigma>0$, such that
    \begin{align}
    \mathbb{E}\big[\exp(\lambda (W\!-\mathbb{E}[W]))\big]\, \leq\, \exp\!{\left(\frac{\lambda^2\sigma^2}{2}\right)}.
    \end{align}
Note that the quantity $\sigma^2$ is often referred to as the variance proxy. Naturally, any Gaussian random variable with variance $\sigma^2$ is $\sigma$-subgaussian.
Furthermore, any distribution with a bounded support on $[a,b]$ is subgaussian with $\sigma = (b-a)/2$ \cite[Example 2.4]{wainwright_2019}. Standard examples of such distributions include the uniform, the triangular and the beta distributions.

Recall that random variable $W$ is called {\em symmetric} about zero if\, $W$ has the same probability distribution as $-W$.
Finally, a random variable $W$ is called {\em nonatomic}, if for all constant $w \in \mathbb{R}$, it holds that $\mathbb{P}(W = w) = 0.$
Notice that every continuous probability distribution is nonatomic.

Our assumptions on the noises are rather mild, since the subgaussian assumption covers a wide range of possible distributions, furthermore, the noise sequence can even be nonstationary, i.e., their distributions can change over time.
The nonatomicity of the distributions is only assumed to simplify the analysis, by using proper random tie-breaking in the SPS Outer Approximation algorithm (Algorithm \ref{alg:sps_eoa}), such as step 3 of Algorithm \ref{alg:sps_indicator}, the assumption could be relaxed.
\begin{assumption}{2}\label{assu:R_nonsigular}
The regressor vectors are ``completely exciting'' in the sense that any $d$ regressors span the whole space, $\BR^d$. 
\end{assumption}
This means that for any subset\, $T$ of the index set $[n]$ with $|\hspace{0.4mm}T\hspace{0.2mm}|=d$, hence, having cardinality $d$, we have that
\vspace{-1mm}
\begin{equation}
    \det\!\left(\frac{1}{d}\sum_{t \in T}\varphi_t\varphi_t\tr \right)\,\neq\, 0.
\end{equation}
\ref{assu:R_nonsigular} ensures under suitable assumptions on the perturbations (step 3 of Algorithm \ref{alg:sps_init}), that the SPS confidence regions, including the outer approximation regions, are {\em bounded}.
\begin{assumption}{3}\label{assu:R_min_eigval}
The excitations are nonvanishing in the sense that there exists a constant\, $\lambda_0 > 0$, such that for all\, $n \geq d:$
\begin{align}
     \lambda_{\text{min}}(\bar{R}_n)\, \geq\, \lambda_0\, > \,0,
\end{align}
where\, $\lambda_{\text{min}}$ denotes the smallest eigenvalue.
\end{assumption}
Assumption \ref{assu:R_min_eigval} guarantees that the averaged ``magnitude'' of the excitation does not get too small, 
hence it provides a lower bound on the signal-to-noise ratio.
\begin{assumption}{4}\label{assu:Phi_Q_coherence}
Let\, $\Phi_n = \Phi_{{\scriptscriptstyle Q},n}\Phi_{{\scriptscriptstyle R},n}$ be the thin QR-decomposition of\, $\Phi_n$. 
There are constants $\kappa >0$ and\, $0 <\rho \leq 1$, such that the following upper bound holds for all\, $n \geq d,$ we have
\begin{align}
    \mu(\Phi_n)\, \defeq\, \frac{n}{d} \max_{1\leq i \leq n}\norm{\Phi_{{\scriptscriptstyle Q},n}\tr e_i}^2 \leq\, \kappa\, n^{1-\rho},
\end{align}
 where $\mu(\Phi_{n})$ is called the coherence of\, $\Phi_{n}$.
\end{assumption}

\ref{assu:Phi_Q_coherence} applies the definition of coherence from \cite[Definition 1.2]{Candes2008ExactMC} along with the facts that $\text{range}(\Phi_n) = \text{range}(\Phi_{{\scriptscriptstyle Q},n})$ and $\Phi_{{\scriptscriptstyle Q},n}\Phi_{{\scriptscriptstyle Q},n}\tr$ is an orthogonal projection onto $\text{range}(\Phi_n)$.
A consequence of this 
is that 
$1 \leq \mu(\Phi_n) \leq n/d$ \cite{Candes2008ExactMC}.
This assumption 
provides an upper bound on how much of the excitation can be ``concentrated'' to a single time step, i.e., onto a particular regressor $\varphi_t$.
In the presented theorems \ref{assu:Phi_Q_coherence} will be assumed alongside with \ref{assu:R_nonsigular}, hence, the thin QR-decomposition of $\Phi_n = \Phi_{{\scriptscriptstyle Q},n}\Phi_{{\scriptscriptstyle R},n}$ is unique, since matrix $\Phi_{{\scriptscriptstyle R},n}$ is full rank.

If the regressors are realizations of i.i.d. continuous random vectors with a positive definite covariance matrix, then it can be shown that the aforementioned assumptions are almost surely satisfied. More rigorously, if the distribution is continuous, then it does not concentrate to any proper affine subspace, hence \ref{assu:R_nonsigular} is satisfied, furthermore, since the regressors are i.i.d., the excitations (a.s.) do not vanish (\ref{assu:R_min_eigval}) and (a.s.) satisfy the coherence assumption (\ref{assu:Phi_Q_coherence}) with some $\rho$ and $\kappa$, as well. 

Realizations of i.i.d. continuous variables as input or as regressors are very common in signal processing practice. Typical examples are the continuous white noise or filtered white noise, with a $d$-ranked matrix FIR or IIR filter.

The SPS EOA algorithm, presented as Algorithm \ref{alg:sps_eoa}, requires much milder assumptions on the noises and regressors.
The introduced stronger assumptions are needed to analyze the 
size of the constructed confidence ellipsoids for any finite sample size 
and almost all realizations of the regressor vectors.

\section{The Sign-Perturbed Sums method}\label{sec:SPS_overview}
This section overviews the original SPS algorithm and its ellipsoidal outer approximation. The intuitive idea behind SPS, detailed descriptions of the algorithms and the proofs of the presented theorems can be found in \cite{Csaji2015} and \cite{szentpeteri2025}.
\subsection{The SPS-Indicator algorithm}\label{sec:SPS-Indicator}
The SPS method consists of two parts. The first part is the initialization, where given the desired confidence probability $p$, the algorithm computes the global variables and generates the required random objects, including the random signs. The SPS-Initialization is given in Algorithm \ref{alg:sps_init}. The second part is the indicator function, which decides whether the input parameter $\theta$ is included in the confidence region. The SPS-Indicator is presented in Algorithm \ref{alg:sps_indicator}. Using this construction, the $p$-level SPS confidence region can be defined as
\begin{equation}
        \confreg[p,n]\, \defeq\, \{\,\theta \in \BR^d\text{ : SPS-Indicator}(\theta) = 1\,\}.
\end{equation}

\floatstyle{spaceruled}
\restylefloat{algorithm}

\begin{algorithm}[t]
    \vspace*{1mm}
	\caption{Pseudocode: SPS-Initialization\hspace{0.5mm}($p$)}
    \label{alg:sps_init}
	\begin{algorithmic}[1]
		\STATE Given the (rational) confidence probability $p \in (0,1)$, set integers $m > q >0$ such that $p = 1 - q/m$.
		\STATE Calculate the outer product $\bar{R}_n$ and find the principle square root $\bar{R}_n^{1/2}$, such that\vspace{-1mm}
		\begin{equation*}
            \bar{R}_n^{1/2}\bar{R}_n^{{1/2}} = \bar{R}_n.
        \end{equation*}
        \vspace{-5mm}
        \STATE Generate $n\cdot(m-1)$ i.i.d. random signs $\{\alpha_{i,t}\}$ for $i \in [\hspace{0.3mm}m-1\hspace{0.2mm}]$ and $t \in [\hspace{0.3mm}n\hspace{0.2mm}]$, with:
        \vspace{-1mm}
        \begin{align*}
            \BP(\alpha_{i,t} = 1) \,=\, \BP(\alpha_{i,t} = -1)\, =\, 1/2.
        \end{align*}
        \vspace{-5mm}
		\STATE Generate (uniformly) a random permutation $\pi$ of the set $\{0,\dots, m - 1\}$, where each of the $m!$ possible permutations has probability $1/(m!)$ to be selected.
	\end{algorithmic}
\end{algorithm}
\begin{algorithm}[t]
	\caption{Pseudocode: SPS-Indicator\hspace{0.5mm}$(\theta)$}
    \label{alg:sps_indicator}
	\begin{algorithmic}[1]
		\STATE Compute the prediction errors for $\theta:$ for $t \in [\hspace{0.3mm}n\hspace{0.2mm}]$ let
        \vspace{-1mm}
        \begin{align*}
            \varepsilon_t(\theta)\, \defeq\, Y_t - \varphi_t\tr \theta.
        \end{align*}
        \vspace{-5mm}
		\STATE Evaluate for $i \in [\hspace{0.3mm}m-1\hspace{0.2mm}]$ the following functions:
        \vspace{-1mm}
        \begin{align*}
             &S_0(\theta)\, \defeq\, \bar{R}_n^{-\frac{1}{2}}\frac{1}{n}\sum_{t=1}^{n}\varphi_t\varepsilon_t(\theta),\\
             &S_i(\theta)\, \defeq\, \bar{R}_n^{-\frac{1}{2}}\frac{1}{n}\sum_{t=1}^{n}\alpha_{i,t}\varphi_t\varepsilon_t(\theta).
        \end{align*}
        \vspace{-1mm}
		\STATE Compute the rank 
        of $\norm{S_0(\theta)}^2$ among $\{\norm{S_i(\theta)}^2\}:$
            \begin{equation*}
                \mathcal{R}(\theta)\, \defeq\, \left[\,1+\sum_{i=1}^{m-1}\mathbb{I}\left(\norm{S_0(\theta)}^2 \succ_{\pi} \norm{S_i(\theta)}^2\right)\,\right]\!,
            \end{equation*}
        where ``$\succ_{\pi}$'' is ``$>$'' with random tie-breaking, i.e., 
        $\norm{S_k(\theta)}^2 \succ_{\pi} \norm{S_j(\theta)}^2$ if and only if $(\norm{S_k(\theta)}^2 > \norm{S_j(\theta)}^2) \lor
        (\norm{S_k(\theta)}^2 = \norm{S_j(\theta)}^2 \,\land\, \pi(k) > \pi(j))$.           
        \vspace{2mm}
 		\STATE Return 1 if\, $\mathcal{R}(\theta) \leq m - q$, otherwise return 0.
	\end{algorithmic}
\end{algorithm}
It is proved in \cite{Csaji2015} that the confidence region $\confreg[p,n]$ covers the true parameter $\theta^*$ with exact probability $p$, under general assumptions on the noise and regressor sequences, which are much milder than \ref{assu:noise}-\ref{assu:Phi_Q_coherence}, hence the following theorem holds.
\begin{theorem}\label{thm:exact_confidence}
    Assuming the noise sequence $\{W_t\}$ contains independent random variables that are symmetric about zero and matrix $\bar{R}_n$ is nonsingular, the coverage probability of the constructed confidence region is exactly $p$, that is,
    \begin{equation}
            \BP(\theta^* \in \confreg[p,n]) = 1-\frac{q}{m} = p.
            \vspace{1mm}
    \end{equation}
\end{theorem}
It has been also rigorously proven that these confidence regions are strongly consistent \cite{Weyer2017}, which requires some further assumptions on the regressor sequence.
The sample complexity of the SPS-Indicator algorithm was studied in \cite{szentpeteri2025}. The following theorem provides high probability upper bounds for the diameters of the confidence regions generated by the SPS-Indicator algorithm for finite sample sizes $n$.
\begin{theorem}\label{thm:sample_complex_indicator}
    Assume \ref{assu:noise}, \ref{assu:R_nonsigular}, \ref{assu:R_min_eigval} and \ref{assu:Phi_Q_coherence}. Then, the following concentration inequality holds for the diameters of SPS sets.
    For all
    $\delta \in (0,1)$ and\, $n \geq \lceil g^{1/\rho}(\frac{\delta}{m-q}) \rceil,$ we have
    \begin{equation}
        \!\!\!\sup_{\theta_1,\theta_2 \in \confreg[p,n]}\!\!\!\!\|\,\theta_1-\theta_2\,\| \,\leq\,
        \dfrac{4\,f\!\left(\frac{\delta}{m-q}\right)}{\left(n^{1-\rho}\lambda_0\!\left(n^{\rho}-g\!\left(\frac{\delta}{m-q}\right)\right)\right)^{\frac{1}{2}}},
        \vspace{1mm}
    \end{equation}
    with probability at least\, $1-\delta,$ where
    \begin{align}\label{equ:f_and_g_not_appendix}
            &f(\delta) \,\defeq\, \notag
        \begin{cases}
            \sigma\hspace{0.3mm}(\hspace{0.3mm}8\hspace{0.5mm}d\ln^{\frac{1}{2}}(\tfrac{4}{\delta})+d)^{\frac{1}{2}} & 4e^{-(nd\lambda_0)^2} \leq \delta \leq \newtext{1},\\[1mm]
            \sigma\hspace{-0.3mm}\left(8\ln(\tfrac{4}{\delta})+d\right)^{\frac{1}{2}} & 0 < \delta < 4e^{-(nd\lambda_0)^2}, \notag\\
        \end{cases}\\[1mm]
        &g(\delta) \,\defeq\, \ln\left(\tfrac{4d}{\delta}\right)2\hspace{0.3mm}\kappa\hspace{0.3mm} d^2.
    \end{align}
\end{theorem}

\subsection{The SPS ellipsoidal outer approximation}
Observe that the SPS-Indicator algorithm detailed in Section \ref{sec:SPS-Indicator} only checks whether a given parameter $\theta$ is included in the confidence region. Building a confidence region using this indicator function requires to evaluate every point on a grid, which is computationally demanding. To give a compact representation of the confidence region around the least squares estimate (LSE) that can be efficiently computed, an ellipsoidal outer approximation (EOA) method was developed, see \cite{Csaji2015}. The confidence region given by this outer approximation is
\begin{equation}\label{equ:SPS_EOA_0onf}
    \confreg[p,n]\, \subseteq\, \outappr[p,n]\, \defeq \bigl\{\hspace{0.3mm}\theta \in \mathbb{R}^d : (\theta - \hat{\theta}_n)\tr  {\bar{R}_n} (\theta - \hat{\theta}_n) \leq r\hspace{0.3mm}\bigl\},
    \vspace{1mm}	
\end{equation}
where $\hat{\theta}_n \defeq R_n^{-1}\Phi_n\tr y_n$ is the LSE and $r$ can be computed from the solutions of the following optimization problems:
\begin{equation}\label{equ:SPS_EOA}
        \begin{aligned}
                \max_{\theta} \quad & \norm{S_0(\theta)}^2\\
                \textrm{s.t.} \quad &\norm{S_0(\theta)}^2 - \norm{S_i(\theta)}^2 \leq 0, \\
        \end{aligned}
\end{equation}
$i \in [\hspace{0.3mm}m-1\hspace{0.2mm}]$.
The precise computation of $r$ and the SPS outer approximation are detailed in Algorithm \ref{alg:sps_eoa}. The optimization problems \eqref{equ:SPS_EOA} are not convex in general, however they can be reformulated to convex semidefinite programs (SDP), which can be solved in polynomial time \cite{Csaji2015}. A detailed description of the these SDP reformulations are given in Section \ref{sec:sample_complex_OA}. As $\outappr[p,n]$ is the sought outer approximation, it satisfies that
\begin{align}
    \BP(\theta^* \in \outappr[p,n]) \geq 1-\frac{q}{m} = p.
\end{align}
\begin{algorithm}[t]
	\caption{Pseudocode: SPS Outer Approximation}
    \label{alg:sps_eoa}
	\begin{algorithmic}[1]
        \STATE  Compute the least-squares estimate,\vspace{-1mm}
        \begin{align*}
            \hat{\theta}_n = R_n^{-1}\Phi_n\tr y_n.
        \end{align*}
        \vspace{-5mm}
		\STATE For $i \in [\hspace{0.3mm}m-1\hspace{0.2mm}]$, solve the optimization problem \eqref{equ:SPS_EOA}, and let $\gamma_i^*$ be the optimal value.
		\STATE Let $r$ be the $q$th largest $\gamma_i^*$ value.
		\STATE The outer approximation of the SPS confidence
            region is given by the ellipsoid 
            \begin{equation*}
                    \outappr[p,n]\, \defeq \bigl\{\hspace{0.3mm}\theta \in \mathbb{R}^d : (\theta - \hat{\theta}_n)\tr  {\bar{R}_n} (\theta - \hat{\theta}_n) \leq r\hspace{0.3mm}\bigl\}.
            \end{equation*}
	\end{algorithmic}
\end{algorithm}
\section{Sample complexity of the\\
SPS Ellipsoidal Outer Approximation}\label{sec:sample_complex_OA}
The SPS ellipsoid gives a compact representation of the confidence region, and compared to applying the SPS-Indicator, it can be used much more easily in various applications, from robust optimization to risk management. It is also significantly less demanding from a computational point of view, as it can be computed in polynomial time. 
The following theorem formalizes the sample complexity of these ellipsoids,
i.e., it shows 
how their radii
shrink as the sample size increases. 
\begin{theorem}\label{thm:sample_complex_eoa}
    Assuming \ref{assu:noise}, \ref{assu:R_nonsigular}, \ref{assu:R_min_eigval} and \ref{assu:Phi_Q_coherence}, the following concentration inequality holds for the sizes of SPS ellipsoidal outer approximation confidence regions.
    For all\, $\delta \in (0, 1)$ and\, $n \geq \lceil g^{1/\rho}(\frac{\delta}{m-q}) \rceil$ with probability at least\, $1-\delta,$ we have
    \begin{align}
        &\sup_{\theta \in \outappr[p,n]}\|\theta - \hat{\theta}_n\| \leq 
        \dfrac{2f(\frac{\delta}{m-q})(n^{\frac{\rho}{2}}+g^{\frac{1}{2}}(\frac{\delta}{m-q}))^{\frac{1}{2}}}{\left(n\lambda_0(n^{\frac{\rho}{2}}-g^{\frac{1}{2}}(\frac{\delta}{m-q}))\right)^{\frac{1}{2}}}.
    \end{align}
\end{theorem}
Notice that unlike in Theorem \ref{thm:sample_complex_indicator}, the size of the region is expressed as the distance between any point in the region and the LSE, which is the center of the outer ellipsoid, see \eqref{equ:SPS_EOA_0onf}. 
The functions $f(\delta)$ and $g(\delta)$ are defined in \eqref{equ:f_and_g_not_appendix}, and, as in Theorem \ref{thm:sample_complex_indicator}, they are independent of the sample size, $n$.

Comparing the sample complexities of the SPS-Indicator and SPS EOA, it can be observed from Theorems \ref{thm:sample_complex_indicator} and \ref{thm:sample_complex_eoa} that the SPS EOA is more conservative, however, the difference is not significant. This behavior is a consequence of the construction, it is expected that an outer approximation has a slightly worse sample complexity. Despite this, as shown in the following corollary, the shrinkage rates of both algorithms are the same, which makes the SPS EOA optimal, since apart from constant factors, this rate cannot be improved.
\begin{corollary}\label{cor:eoa_rate}
    Under the assumptions of Theorem \ref{thm:sample_complex_eoa}, the sizes of the confidence regions generated by the SPS ellipsoidal outer approximation algorithm shrink at a rate $\mathcal{O}(1/\sqrt{n})$.
\end{corollary}
\begin{proof}
    From the sample size $(n)$ dependent terms in the upper bound of Theorem \ref{thm:sample_complex_eoa}, it can be seen that the decrease rate is $\mathcal{O}(n^{\frac{\rho}{4}} / (n^{\frac{1}{2}}\cdot n^{\frac{\rho}{4}})) = \mathcal{O}(1/\sqrt{n})$.
\end{proof}
\begin{remark}
In our previous work \cite{szentpeteri2023scalar}, we have investigated the sample complexity of the SPS EOA for the special case of scalar linear regression. That result concluded geometric decrease rate for the sizes of the SPS EOA intervals, however Corollary \ref{cor:eoa_rate} is not only more general, but also improves upon that rate. The reason behind this 
is that in \cite{szentpeteri2023scalar} the derivation of the 
EOA concentration bounds used the same techniques as the proof of high probability upper bounds on the distance between the confidence interval and the true parameter $\theta^*$. Note that in Theorems \ref{thm:sample_complex_indicator} and \ref{thm:sample_complex_eoa} the size of the confidence set is not expressed as a distance from the true parameter, but as a distance between any two points in the region.
\end{remark}

\begin{proof}[\textbf{Proof of Theorem \ref{thm:sample_complex_eoa}}]
In the first step of the proof, we assume that $m=2$ and $q=1$, then we reformulate the optimization problem \eqref{equ:SPS_EOA} into a convex semidefinite program, unless the region is unbounded which case is treated separately. We then show that the SDP has an optimal value, and also give an upper bound for this optimal value. In the final step, we derive a concentration inequality for this upper bound (and for the size of the confidence region), and generalize this result to arbitrary $m$ and $q$ choices. 

{\em Step i)} In order to obtain an outer approximation of the confidence region for two sums ($m=2$), the following optimization problem should be solved \eqref{equ:SPS_EOA}
\vspace{-1mm}
\begin{equation}\label{equ:SPS_EOA_m2}
    \begin{aligned}
        \max_{\theta} \quad & \norm{S_0(\theta)}^2\\
        \textrm{s.t.} \quad & \norm{S_0(\theta)}^2 - \norm{S_1(\theta)}^2\leq 0.
    \end{aligned}
\end{equation}
To reformulate the above optimization problem as a convex problem, a result from \cite{Csaji2015} can be used. 
\vspace{-1mm}
Let
\begin{equation}
    \psi_n \defeq \Phi_n\tr  D_{\alpha,n} y_n.
\end{equation}
By expanding $\norm{S_0(\theta)}^2$ as in \cite[Section VI.]{Csaji2015}, we get that
\begin{equation}\label{equ:reform_obj_cvx}
    \norm{S_0(\theta)}^2 = (\theta - \hat{\theta}_n)\tr  \bar{R}_n (\theta - \hat{\theta}_n).
\end{equation}
Then, by introducing $z \defeq \bar{R}_n^{\frac{1}{2}}(\theta - \hat{\theta}_n)$, the optimization problem \eqref{equ:SPS_EOA_m2} can be reformulated as \cite[Section VI.]{Csaji2015}
\begin{equation}\label{equ:SPS_EOA_m2_z}
    \begin{aligned}
        \max \quad & \norm{z}^2\\
        \textrm{s.t.} \quad & z\tr  A_0 z + 2z\tr  b_0 + c_0 \leq 0,
    \end{aligned}
    \vspace{-2mm}
\end{equation}
where
\begin{align}
&A_0 &\defeq& I - R_n^{-\frac{1}{2}}Q_nR_n^{-1}Q_nR_n^{-\frac{1}{2}}= R_n^{-\frac{1}{2}}AR_n^{-\frac{1}{2}},\label{equ:def_Ac}\\
&b_0 & \defeq &\frac{1}{\sqrt{n}}R_n^{-\frac{1}{2}}Q_nR_n^{-1}\left(\psi_n - Q_n\hat{\theta}_n\right)\notag\\
    && = &\frac{1}{\sqrt{n}}R_n^{-\frac{1}{2}}Q_nR_n^{-1}(\Phi_n\tr  D_{\alpha,n} - Q_nR_n^{-1}\Phi_n\tr )y_n\notag\\
    && = &\frac{1}{\sqrt{n}}R_n^{-\frac{1}{2}}Q_nR_n^{-1} B D_{\alpha,n} y_n,\label{equ:def_bc}\\
&c_0 &\defeq&\frac{1}{n}\left(-\psi_n\tr  R_n^{-1}\psi_n + 2\hat{\theta}_n\tr  Q_nR_n^{-1}\psi_n - \hat{\theta}_n\tr                            Q_nR_n^{-1}Q_n\hat{\theta}_n\right)\notag\\
    && = &-\frac{1}{n}\left(y_n\tr  \left[D_{\alpha,n}\Phi_n R_n^{-1}\Phi_n\tr  D_{\alpha,n}  \right.\right.\notag\\
    &&&\left.\left. -\Phi_n R_n^{-1}Q_nR_n^{-1}\Phi_n\tr  D_{\alpha,n} - D_{\alpha,n}\Phi_n R_n^{-1}Q_nR_n^{-1}\Phi_n\tr  \right.\right.\notag\\
    &&&\left.\left. + \Phi_n R_n^{-1}Q_nR_n^{-1}QR_n^{-1}\Phi_n\tr \right]y_n\right)\notag\\
    &&=&-\frac{1}{n}(y_n\tr( D_{\alpha,n}\Phi_n - \Phi_n R_n^{-1}Q_n)R_n^{-1}(\Phi_n\tr D_{\alpha,n} - \notag\\ &&&Q_nR_n^{-1}\Phi_n\tr)y_n)\notag\\
    &&=&-\frac{1}{n}(y_n\tr D_{\alpha,n}B\tr R_n^{-1}B D_{\alpha,n}y_n),
\end{align}
and $D_{\alpha, n}$, $Q_n$, $A$ and $B$ are defined as
\begin{align}\label{equ:def_D}
    &D_{\alpha, n} \defeq
    \begin{bmatrix}
        \alpha_{1} & & \\
        & \ddots & \\
        & & \alpha_{n}
    \end{bmatrix},\\
    &Q_n \defeq \sum_{t=1}^n \alpha_t\varphi_t\varphi_t\tr  = \Phi_n\tr  D_{\alpha,n}\Phi_n,\label{equ:def_Q}\\
    &A \defeq R_n - Q_nR_n^{-1}Q_n,\label{equ:def_A}\\ 
    &B \defeq \Phi_n\tr - Q_nR_n^{-1}\Phi_n\tr  D_{\alpha,n}.\label{equ:def_B}
\end{align}
Note that in \cite[Theorem 1]{Care2022} it has been shown that the SPS confidence regions are bounded, if both the perturbed and the unperturbed regressors span the whole space. This means that assuming \ref{assu:R_nonsigular} and $m=2$, 
the SPS region is bounded if $A$ is positive definite, which is independent of $\{W_t\}$.

In case of a bounded SPS confidence region, the ellipsoid $\mathcal{E}_1 = \{\theta: \|S_0(\theta)\|^2 - \|S_1(\theta)\|^2 \leq 0 \}$ is compact and non-empty, since it is a closure of the bounded SPS-Indicator set $\{\theta: \|S_0(\theta)\|^2 \prec_{\pi} \|S_1(\theta)\|^2 \}$, and by construction of the SPS region $\|S_0(\hat{\theta}_n)\|^2 = 0$, therefore the LSE is always included in $\mathcal{E}_1$. In Lemma \ref{lemma:S1_LS_notnull} we show that if the SPS region is bounded, given $\{\varphi_t\}$ and $\{\alpha_{1,t}\}$, it almost surely holds that $\|S_1(\hat{\theta}_n)\|^2 \neq 0$. From the previous two properties, it follows that the optimization problem \eqref{equ:SPS_EOA_m2} is strictly feasible, i.e., Slater's condition hold, if we work with bounded SPS confidence regions. Consequently, strong duality holds \cite[Appendix B]{Boyd2009}, therefore the value of the above optimization program \eqref{equ:SPS_EOA_m2_z} is equal to the value of its dual which can be formulated by using the Schur complement, as follows:
\vspace{-1mm}
\begin{equation}\label{equ:sps_eoA_0vxopt}
\begin{aligned}
    \min \quad & \gamma_0\\
    \textrm{s.t.} \quad & \xi \geq 0 \\
                        & \begin{bmatrix}
                            -I + \xi A_0 & \xi b_0\\
                            \xi b_0\tr  & \xi c_0 + \gamma_0
                        \end{bmatrix} \succeq 0.
\end{aligned}
\vspace{1mm}
\end{equation}
The semidefinite optimization problem \eqref{equ:sps_eoA_0vxopt} is {\em convex} and it can be solved using standard convex optimization tools. 
    {\newtext 
    Note that problem \eqref{equ:sps_eoA_0vxopt} was originally proposed by \cite{Csaji2015}, but only for the case of bounded regions. Nevertheless, here we also allow and handle the case of unbounded regions, c.f.\ \eqref{equ:gammas_def}.
    }

{\em Step ii)} As a next step, we argue that there exists an
optimal value $\gamma_0^*$ for the semidefinite program \eqref{equ:sps_eoA_0vxopt}, and also derive a formula for the solution of \eqref{equ:SPS_EOA_m2}, denoted by $\gamma^*$, which  is even valid for unbounded regions. Note that when we investigate the properties of $\gamma_0^*$ or that of problem \eqref{equ:sps_eoA_0vxopt}, we always assume that the confidence region is bounded, see \eqref{equ:gammas_def}.

In case of unbounded SPS regions, 
the solution of program \eqref{equ:SPS_EOA_m2} is always $\infty$. To construct a well-defined solution $\gamma^*$ for the optimization problem
\eqref{equ:SPS_EOA_m2},
we introduce the
cases
\begin{align}\label{equ:gammas_def}
    \gamma^* \defeq
    \begin{cases}
        \gamma_0^* &\text{if the SPS region is bounded,}\\
        \infty & \text{if the SPS region is unbounded.} \\
    \end{cases}
\end{align}

Since $\norm{S_0(\theta)}^2$ is a continuous function and in \eqref{equ:SPS_EOA_m2} the constraint set is compact in the case of bounded SPS regions, it follows that \eqref{equ:SPS_EOA_m2} has a solution. As we mentioned, in the reformulation from \eqref{equ:SPS_EOA_m2} to \eqref{equ:sps_eoA_0vxopt} strong duality holds, hence the convex optimization problem \eqref{equ:sps_eoA_0vxopt} has an optimal value $\gamma_0^*$.

Throughout the proof we will use the following notation
\begin{equation}\label{equ:def_Ac-coma}
    A'_0(\xi) \defeq -\tfrac{1}{\xi}I +  A_0 = -\tfrac{1}{\xi}I + R_n^{-\frac{1}{2}}AR_n^{-\frac{1}{2}},
\end{equation}
\begin{equation}
    Z \defeq \begin{bmatrix}
        \xi A'_0(\xi) &  \xi b_0 \\
         \xi b_0\tr  &  \xi c_0 + \gamma_0
    \end{bmatrix}.
\end{equation}
Lower bounds on
variables 
$\gamma_0$ and $\xi$
\eqref{equ:sps_eoA_0vxopt} can be determined by the conditions for positive semidefiniteness 
of (generalized) Schur complements \cite[Appendix A.5.5]{Boyd2009}, that is
\begin{numcases}{Z \succeq 0 \Leftrightarrow}
\xi A'_0(\xi) \succeq 0,\label{equ:psd_condition_1}\\\notag\\
\left(I-\left(\xi A'_0(\xi)\right)\left(\xi A'_0(\xi)\right)^\dagger\right)\xi b_0 =\notag\\
\left(I-A'_0(\xi)\left(A'_0(\xi)\right)^\dagger\right)\xi b_0 = 0,\label{equ:psd_condition_2}\\\notag\\
\xi c_0 + \gamma_0 - \xi b_0\tr  \left(\xi A'_0(\xi)\right)^\dagger \xi b_0 = \notag\\
\xi c_0 + \gamma_0 - \xi b_0\tr  \left(A'_0(\xi)\right)^\dagger b_0 \succeq 0,\label{equ:psd_condition_3}
\end{numcases}
where we used the notation ``$(\cdot)^\dagger$'' for the pseudoinverse. 

To give a lower bound on $\xi$, we will reformulate the condition $\xi A'_0(\xi) \succeq 0$. First, we rewrite $A$ {\newtext as in \cite{szentpeteri2025}}
\begin{align}\label{equ:reformulation_A}
    A = &\;R_n - Q_nR_n^{-1}Q_n = \Phi_{R,n}\tr \Phi_{R,n} - \notag\\
    &\Phi_{R,n}\tr (\Phi_{Q,n}\tr  D_{\alpha,n}\Phi_{Q,n}\Phi_{Q,n}\tr  D_{\alpha,n}\Phi_{Q,n})\Phi_{R,n},
\end{align}
where $\Phi_n = \Phi_{Q,n}\Phi_{R,n}$ is the thin QR-decomposition of $\Phi_n$, and $\Phi_{R,n}$ is nonsingular ensured by \ref{assu:R_nonsigular}.
By introducing 
\begin{align}\label{equ:def_K}
    &K \defeq \Phi_{Q,n}\tr  D_{\alpha,n}\Phi_{Q,n},
\end{align}
matrix $A$ can be further reformulated as
\begin{align}\label{equ:reformulation_A_2}
    &\Phi_{R,n}\tr (I -\Phi_{Q,n}\tr  D_{\alpha,n}\Phi_{Q,n}\Phi_{Q,n}\tr  D_{\alpha,n}\Phi_{Q,n})\Phi_{R,n} = \notag\\ 
    &\Phi_{R,n}\tr (I - K^2)\Phi_{R,n}.
\end{align}
Note that matrix $A$ is positive semidefinite 
and $\Phi_{R,n}$ is full rank (\ref{assu:R_nonsigular}), therefore $I-K^2$ is also positive semidefinite. Thus, it holds for the eigenvalues of $K$ that $\forall\, i: \abs{\lambda_i(K)} \leq 1$. 

Recall that the SPS set is unbounded if and only if $A$ has zero eigenvalue \cite{Care2022}.
It follows that $K^2$ has an eigenvalue $1$, if and only if the region is unbounded. From \ref{assu:R_nonsigular} it follows that $K$ is nonsingular, therefore $\forall\, i:0 < \abs{\lambda_i(K)}$. Substituting formula \eqref{equ:reformulation_A_2} for $A$ back to the first condition \eqref{equ:psd_condition_1}, we get
\begin{align}\label{equ:lam_lb_start}
    &-I + \xi(R_n^{-\frac{1}{2}}\Phi_{R,n}\tr (I-K^2) \Phi_{R,n}R_n^{-\frac{1}{2}}) = \notag\\
    &-I + \xi(I -R_n^{-\frac{1}{2}}\Phi_{R,n}\tr K^2 \Phi_{R,n}R_n^{-\frac{1}{2}}) \succeq 0.
\end{align}
The principal square root matrix $R_n^{-1/2}$ can be reformulated as $V_{\Phi_{R}}\Sigma_{\Phi_{R}}^{-1}V_{\Phi_{R}}\tr $ using the SVD decomposition of $\Phi_{R,n} = U_{\Phi_{R}}\Sigma_{\Phi_{R}}V_{\Phi_{R}}\tr $, since $R_n = \Phi_{R,n}\tr \Phi_{R,n} = V_{\Phi_{R}}\Sigma_{\Phi_{R}}^2V_{\Phi_{R}}\tr $. Then, the first condition \eqref{equ:psd_condition_1} can be written as
\begin{align}\label{equ:Phi_R_SVD}
    &-I + \xi(I - V_{\Phi_{R}}U_{\Phi_{R}}\tr K^2U_{\Phi_{R}}V_{\Phi_{R}}\tr ) \succeq 0.
\end{align}
Using that $ V_{\Phi_{R}}$ and $U_{\Phi_{R}}$ are orthogonal matrices, we have $V_{\Phi_{R}}U_{\Phi_{R}}\tr  U_{\Phi_{R}}V_{\Phi_{R}}\tr  = I$, the condition can be rewritten as
\begin{equation}\label{equ:A_psi_reform}
V_{\Phi_{R}}U_{\Phi_{R}}\tr (-I + \xi(I -K^2)) U_{\Phi_{R}}V_{\Phi_{R}}\tr  \succeq 0,
\end{equation}
from which it follows that
\begin{align}\label{equ:lam_lb_end}
    &\xi(I -K^2)-I \succeq 0, \notag\\
    &\xi \geq \frac{1}{\lambda_{\text{min}}(I-K^2)}.
\end{align}
Since we know that 
$\forall\, i:0<\abs{\lambda_i(K)} < 1$ (cf.\ bounded SPS region), it follows that $1/\lambda_{\text{min}}(I-K^2) =1/(1-\lambda_{\text{max}}(K^2))$.

We conclude from the first condition \eqref{equ:psd_condition_1} that the minimum value that $\xi$ can attain is $1/(1-\lambda_{\text{max}}(K^2))$. Lemma \ref{lemma:condition2_not_null} shows that if $\xi = 1/(1-\lambda_{\text{max}}(K^2))$, then the second condition \eqref{equ:psd_condition_2} is violated, i.e., $\big(I-\left( A'_0(\xi) \right)\left(A'_0(\xi)\right)^\dagger\big)\xi b_0 \neq 0$ almost surely. Hence, it must hold that $\xi > 1/(1-\lambda_{\text{max}}(K^2))$ to satisfy the first and second conditions simultaneously. 
Therefore, $\xi^*$ 
takes the form $\xi^*= 1/(1-\lambda_{\text{max}}(K^2)) + \delta_{\xi}$, for some $\delta_{\xi} > 0$. 

Next, from the third condition \eqref{equ:psd_condition_3}, we can deduce a lower bound on $\gamma_0$ as follows
\begin{align}
    &\xi c_0 + \gamma_0 - \xi b_0\tr \left(A'_0(\xi)\right)^\dagger b_0 \succeq 0 \notag\\
    &\gamma_0 \geq \xi \left(b_0\tr \left(A'_0(\xi)\right)^\dagger b_0 - c_0\right).
\end{align}
So far we have reformulated the constraints of the optimization problem \eqref{equ:sps_eoA_0vxopt}, $\xi \geq 0$ and $Z \succeq 0$, to $\xi >1/(1-\lambda_{\text{max}}(K^2))$ and $\gamma_0 \geq \xi \left(b_0\tr \left(A'_0(\xi)\right)^\dagger b_0 - c_0\right)$. The value of $\gamma_0$ only depends on one constraint, therefore it is straightforward to see that
\begin{equation}
    \gamma_0^* = \xi^* \left(b_0\tr \left(A'_0(\xi^*)\right)^\dagger b_0 - c_0\right),
\end{equation}
since if for some $\xi^* >1/(1-\lambda_{\text{max}}(K^2))$, there were a $\gamma_0^* \neq \xi^* \left(b_0\tr \big(A'_0(\xi^*)\right)^\dagger b_0 - c_0\big)$, then there would be a $(\gamma_0^{*})'$ that satisfies the constraints of the  problem and $(\gamma_0^{*})' < \gamma_0^*$.

The formula for $\gamma_0^*$ can be extended to $\gamma^*$, since as we mentioned before, if the SPS region is unbounded, then $\lambda_{\text{max}}(K) = 1$ and we can consider $\xi^* = $``$1/0$''$ +  \delta_{\xi} = \infty$, therefore $\gamma^* = \infty$ as in \eqref{equ:gammas_def}.

{\em Step iii)} Now that we have a formula for $\gamma^*$ we rewrite and give an upper bound for it that does not depend on $\xi^*$. As in the previous step, we first reformulate $\gamma_0^*$ in the case of bounded SPS regions and then extend it to unbounded regions.

Our secondary objective is to bring this upper bound into a form such that the following two lemmas could be applied, which were originally stated and proved in \cite{szentpeteri2025}.
\begin{lemma}\label{lemma:numerator_indicator}
    Assuming \ref{assu:noise} and \ref{assu:R_nonsigular}, the following concentration inequality holds for 
    $X = w_n\tr  M w_n$, where $M$ is a symmetric projection matrix with rank($M$) $\leq d$: for every $\varepsilon \geq 0$,
    \begin{align}
        \BP\left(\frac{\vert X - \mathbb{E}X\vert}{n\lambda_0} \geq \varepsilon\right) \leq 
        \begin{cases}
            2\exp(-\frac{\varepsilon^2n^2\lambda_0^2}{64d^2\sigma^4}) & 0 \leq \varepsilon \leq 8\sigma^2 d^2\\[1mm]
            2\exp(-\frac{\varepsilon n\lambda_0}{8\sigma^2}) & \varepsilon > 8\sigma^2{d^2}.
        \end{cases}
    \end{align}   
\end{lemma}
\begin{lemma}\label{lemma:denominator_indicator}
    Assuming \ref{assu:R_nonsigular} and \ref{assu:Phi_Q_coherence} the following concentration inequality holds for every $0 < \varepsilon_0 \leq 1$,
    \vspace{-1mm}
    \begin{equation}
        \BP\left(\max_i \abs{\lambda_i(K)} \geq \varepsilon_0\right) \leq 2d\exp\left(-\frac{n^{\rho}\varepsilon_0^2}{2\kappa d^2}\right).
        \vspace{1mm}
    \end{equation}
\end{lemma}

First, by \eqref{equ:def_bc}, the optimal value,
$\gamma_0^*$, can be reformulated as
\begin{align}
    \gamma_0^* &= \xi^* \left(b_0\tr \left(A'_0(\xi^*)\right)^\dagger b_0 - c_0\right) \notag\\
    & =\dfrac{\xi^*}{n}\left(y_n\tr \left[D_{\alpha,n}B\tr R_n^{-1}Q_n R_n^{-\frac{1}{2}}\left(A'_0(\xi^*)\right)^\dagger\right.\right.\\
    &\quad \left.\left.\cdot R_n^{-\frac{1}{2}}Q_n R_n^{-1}BD_{\alpha,n} + D_{\alpha,n}B\tr R_n^{-1}B D_{\alpha,n} \right]y_n\right).\notag
\end{align}
Since, by the definition of $B$, i.e.~\eqref{equ:def_B}, we have
\begin{align}\label{equ:y_to_w_rewrite}
    &BD_{\alpha,n}y_n = (\Phi_n\tr D_{\alpha,n} - Q_n R_n^{-1}\Phi_n\tr)y_n = \notag\\
    &(\Phi_n\tr D_{\alpha,n} - Q_n R_n^{-1}\Phi_n\tr)(\Phi_n\theta^* +w_n) =  \notag\\
    &(Q_n-Q_nR_n^{-1}R_n)\theta^* + (\Phi_n\tr D_{\alpha,n} - Q_n R_n^{-1}\Phi_n\tr)w_n= \notag\\
    &(\Phi_n\tr D_{\alpha,n} - Q_n R_n^{-1}\Phi_n\tr)w_n=BD_{\alpha,n}w_n,
\end{align}
the optimal value of 
problem \eqref{equ:sps_eoA_0vxopt} can be written as
\begin{align}
    &\gamma_0^* = \dfrac{w_n\tr \xi^*P(\xi^*)w_n}{n},
\end{align}
where symmetric matrix $P(\xi^*)$ is defined as 
\begin{align}\label{equ:P_defined}
    P(\xi^*) \defeq\, &D_{\alpha,n}B\tr \Big[R_n^{-1}Q_n R_n^{-\frac{1}{2}}\left(A'_0(\xi^*)\right)^\dagger R_n^{-\frac{1}{2}}Q_n R_n^{-1}\notag\\
    &+R_n^{-1}\Big]BD_{\alpha,n}.
\end{align}

As we mentioned earlier, we aim at giving an upper bound for $\gamma_0^*$ for which we could use the results of Lemma \ref{lemma:numerator_indicator} and Lemma \ref{lemma:denominator_indicator}, therefore in the following we express the eigenvalues of $\xi^*P(\xi^*)$ with the eigenvalues of $K$. 
From Lemma \ref{lemma:eig_P} it follows that in case $\forall\, i:0<\abs{\lambda_i(K)} < 1$ and $\xi^* >1/(1-\lambda_{\text{max}}(K^2))$, $\text{rank}(\xi^*P(\xi^*)) = d$ and the nonzero eigenvalues of $\xi^*P(\xi^*)$  can be written as
\begin{align}\label{equ:eigvals_P}
    \lambda_i(\xi^* P(\xi^*)) = \frac{\xi^*(\xi^*\lambda_i^2(K) - \xi^* - \lambda_i^2(K) + 1)}{\xi^*\lambda_i^2(K) - \xi^* + 1}.
\end{align}

An upper bound on the optimal value $\gamma_0^*$ can be given by using the eigendecomposition of $\xi^*P(\xi^*) = \xi^*V_{\scriptscriptstyle P(\xi^*)}\Lambda_{P(\xi^*)}V_{\scriptscriptstyle P(\xi^*)}\tr$, the fact that rank$\left(\xi^*P(\xi^*)\right) =d$ and introducing $X_{\scriptscriptstyle P} \defeq w_n\tr V_{\scriptscriptstyle P(\xi^*)}D_dV_{\scriptscriptstyle P(\xi^*)}\tr w_n$ as
\begin{align}
    \gamma_0^* &= \dfrac{w_n\tr \xi^*P(\xi^*)w_n}{n}\notag\\
    & \leq \dfrac{\xi^*\lambda_{\text{max}}(P(\xi^*))}{n}\left(w_n\tr V_{\scriptscriptstyle P(\xi^*)}D_dV_{\scriptscriptstyle P(\xi^*)}\tr w_n\right)\notag\\
    & \leq \dfrac{\xi'\lambda_{\text{max}}(P(\xi'))X_{\scriptscriptstyle P}}{n},
\end{align}
where $D_d \defeq \text{diag}(0,\dots,0,1,\dots,1)$ with exactly $d$ ones, and $\xi'$
minimizes $\xi\lambda_{\text{max}}(P(\xi))$ in the form $\xi = 1/(1-\lambda_{\text{max}}(K^2)) + \delta_{\xi}$, $\delta_{\xi} > 0$, such that the upper bound would be the infimum of all possible bounds that satisfies the constraint on $\xi$.

With this construction of $\xi'$ an upper bound for $\gamma_0^*$ can be derived that only depends on $\lambda_{\text{max}}(K^2)$ and not on $\xi'$ the following way. By using Lemma \ref{lemma:eig_P} and substituting $\xi' = 1/(1-\lambda_{\text{max}}(K^2)) + \delta'_{\xi}$ to $\xi\lambda_{\text{max}}(P(\xi))$ we get
\begin{align}
    \xi'\lambda_{\text{max}}(P(\xi')) &=\big[(\delta'_{\xi})^2(\lambda^2_{\text{max}}(K^2) - 2\lambda_{\text{max}}(K^2) + 1) \notag\\
    &\quad-\delta'_{\xi}(\lambda^2_{\text{max}}(K^2) + 1) + \lambda_{\text{max}}(K^2)\big]\notag\\
    &\quad/(\delta'_{\xi}(\lambda^2_{\text{max}}(K^2) - 2\lambda_{\text{max}}(K^2) + 1)).
\end{align}
Consider $\xi'\lambda_{\text{max}}(P(\xi'))$ as a function of $\delta'_{\xi}$ denoted by $f_P(\delta'_{\xi})$.
It holds for the second derivative of the function $f_P(\delta'_{\xi})$ that
\begin{align}
     f^{(2)}_P(\delta'_{\xi}) = \frac{\partial^2f_P(\delta'_{\xi})}{\partial(\delta'_{\xi})^2} = \frac{2 \lambda_{\text{max}}(K^2)}{(\delta'_{\xi})^3(\lambda_{\text{max}}(K^2)-1)^2} > 0,
\end{align}
since $\delta'_{\xi} > 0$ and $0<\abs{\lambda_i(K)}<1$, therefore the function is strictly convex.
By solving the equation
\begin{equation}
    \frac{\partial f_P(\delta'_{\xi})}{\partial\delta'_{\xi}} = 0,
\end{equation} 
it can be obtained that $\xi'\lambda_{\text{max}}(P(\xi'))$ takes its minimum at $ \delta'_{\xi} = \lambda^{1/2}_{\text{max}}(K^2)/(1-\lambda_{\text{max}}(K^2))$, therefore $\xi' = (1+\lambda^{1/2}_{\text{max}}(K^2))/(1-\lambda_{\text{max}}(K^2))$. Then, the optimal value of the program can be upper bounded by substituting the determined $\xi'$ to $\xi\lambda_{\text{max}}(P(\xi))$ using Lemma \ref{lemma:eig_P} as
\begin{align}
    \gamma_0^* &= \dfrac{\xi^*}{n}\left(w_n\tr P(\xi^*)w_n\right)
    \leq \dfrac{\xi'\lambda_{\text{max}}(P(\xi'))X_{\scriptscriptstyle P}}{n}\notag\\
    &\leq \frac{X_{\scriptscriptstyle P}(1+\lambda^{1/2}_{\text{max}}(K^2))}{n(1-\lambda^{1/2}_{\text{max}}(K^2))}.
\end{align}

The upper bound for $\gamma_0^*$ (bounded case) can be extended to $\gamma$ (unbounded case) by considering that $\lambda_{\text{max}}(K) = 1$ and $\gamma_0^* = $``$1/0$''$ = \infty$, therefore $\gamma^* = \infty$ as in \eqref{equ:gammas_def}.

{\em Step iv)} In the last step, we use the upper bound on $\gamma^*$ to derive an upper bound for the size of the confidence region. We then apply the results of Lemmas \ref{lemma:numerator_indicator} and \ref{lemma:denominator_indicator} to give a concentration inequality for the size of the region and generalize the high probability bound for arbitrary $m$ and $q$.

In the case of $m=2$, $q=1$ the confidence region generated by the SPS EOA can be rewritten as
\begin{align}
    &(\theta - \hat{\theta}_n)\tr  \bar{R}_n (\theta - \hat{\theta}_n) \leq \gamma^*,
\end{align}
which follows from Algorithm \ref{alg:sps_eoa}, since $r = \gamma^*$ for one perturbed sum. Using our upper bound on $\gamma^*$ it holds that
\begin{align}
    &\frac{X_{\scriptscriptstyle P}(1+\lambda^{1/2}_{\text{max}}(K^2))}{n(1-\lambda^{1/2}_{\text{max}}(K^2))}\geq (\theta - \hat{\theta}_n)\tr \bar{R}_n (\theta - \hat{\theta}_n) = \notag\\
    &\frac{1}{n} (\theta - \hat{\theta}_n)\tr R_n (\theta - \hat{\theta}_n) = \frac{1}{n} \norm{(\theta - \hat{\theta}_n)\tr R_n^{1/2}}^2\geq\notag\\
    &\frac{1}{n} \lambda_{\text{min}}(R_n)\|\theta - \hat{\theta}_n\|^2,
\end{align}
therefore
\begin{align}
    \|\theta - \hat{\theta}_n\|^2 &\leq \frac{X_{\scriptscriptstyle P}(1+\lambda^{1/2}_{\text{max}}(K^2))}{\lambda_{\text{min}}(R_n)(1-\lambda^{1/2}_{\text{max}}(K^2))}\notag\\
    & = \frac{\tfrac{1}{n}X_{\scriptscriptstyle P}(1+\lambda^{1/2}_{\text{max}}(K^2))}{\tfrac{1}{n}\lambda_{\text{min}}(R_n)(1-\lambda^{1/2}_{\text{max}}(K^2))}.
\end{align}
Using assumption \ref{assu:R_min_eigval} it can be written that
\begin{align}\label{equ:eoa_ellipsoid_upperbound}
    \|\theta - \hat{\theta}_n\|^2 \leq \frac{\tfrac{1}{n}X_{\scriptscriptstyle P}(1+\lambda^{1/2}_{\text{max}}(K^2))}{\lambda_0(1-\lambda^{1/2}_{\text{max}}(K^2))}.
\end{align}

In the following analysis we investigate the two terms $X_{\scriptscriptstyle P}/(n \lambda_0)$ and $(1+\lambda^{1/2}_{\text{max}}(K^2))/(1-\lambda^{1/2}_{\text{max}}(K^2))$ separately.
Notice that in $X_{\scriptscriptstyle P} = w_n\tr V_{\scriptscriptstyle P(\xi^*)}D_dV_{\scriptscriptstyle P(\xi^*)}\tr w_n$, the matrix $V_{\scriptscriptstyle P(\xi^*)}D_dV_{\scriptscriptstyle P(\xi^*)}\tr$ is a projection matrix with rank $d$, since $V_{\scriptscriptstyle P(\xi^*)}$ is an orthonormal (eigenvector) matrix and $D_d$ is a diagonal with $d$ values of 1 and $n-d$ values of 0 in its diagonal. As a consequence the result of Lemma \ref{lemma:numerator_indicator} can be applied to $X_{\scriptscriptstyle P}$
\begin{align}\label{equ:eoa_concentration_firstpart}
    \BP\left(\frac{\vert X_{\scriptscriptstyle P} - \mathbb{E}X_{\scriptscriptstyle P}\vert}{n\lambda_0} \geq \varepsilon\right) \leq 
    \begin{cases}
        2\exp(\frac{-\varepsilon^2n^2\lambda_0^2}{64d^2\sigma^4}) & \hspace{-0.5mm}0 \leq \varepsilon \leq 8\sigma^2 d^2\\[1mm]
        2\exp(\frac{-\varepsilon n\lambda_0}{8\sigma^2}) & \hspace{-0.5mm}\varepsilon > 8\sigma^2{d^2}.
    \end{cases}
\end{align}

In order to give a concentration inequality for the second term, we investigate the probability for every $\varepsilon > 1$,
\begin{align}
    \label{eq:ub_second_term}
    \hspace{-2mm}\BP\left(\frac{1+\lambda^{1/2}_{\text{max}}(K^2)}{1-\lambda^{1/2}_{\text{max}}(K^2)}\geq \varepsilon\right) = \BP\left(\max\abs{\lambda_i(K)} \geq \frac{\varepsilon -1}{\varepsilon +1} \right)\!.
\end{align}
An upper bound for \eqref{eq:ub_second_term}
can be given by using Lemma \ref{lemma:denominator_indicator} as
\vspace{-1mm}
\begin{equation}\label{equ:eoa_concentration_secpart}
\BP\left(\frac{1+\lambda^{1/2}_{\text{max}}(K^2)}{1-\lambda^{1/2}_{\text{max}}(K^2)}\geq \varepsilon\right) \leq 2d\exp\left(-\frac{n^{\rho}\left(\frac{\varepsilon -1}{\varepsilon +1}\right)^2}{2\kappa d^2}\right)\!.
\vspace{-2mm}
\end{equation}

The following lemma gives a concentration inequality for the size of the confidence region generated by the SPS outer approximation algorithm for $m=2$ and $q=1$, by combining the results of \eqref{equ:eoa_concentration_firstpart} and \eqref{equ:eoa_concentration_secpart}.
\begin{lemma}\label{lemma:sample_complex_eoa_m2_q1}
    Assuming \ref{assu:noise}-\ref{assu:Phi_Q_coherence}, the following concentration inequality holds for the sizes of the confidence regions generated by the SPS Outer Approximation algorithm (Algorithm \ref{alg:sps_eoa}) for $m=2$ and $q=1$, with probability at least $1-\delta:$
    \begin{align}
    &\sup_{\theta \in \outappr[0.5,n]}\|\theta - \hat{\theta}_n\| \leq 
        \dfrac{f(\delta)(n^{\frac{\rho}{2}}+g^{\frac{1}{2}}(\delta))^{\frac{1}{2}}}{\left(n\lambda_0(n^{\frac{\rho}{2}}-g^{\frac{1}{2}}(\delta))\right)^{\frac{1}{2}}}.
    \end{align}
\end{lemma}

The proof of Lemma \ref{lemma:sample_complex_eoa_m2_q1} is built upon the same ideas as the
proof of \cite[Lemma 4]{szentpeteri2025}, but since our second term is different, we present it in Appendix \ref{sec:proof_sample_complex_eoa_m2_q1}, for the sake of completeness.

To complete the proof, we consider the general case, i.e., we allow arbitrary $m > q > 0$ (integer) choices. 
From the construction, see Algorithm \ref{alg:sps_eoa}, it is clear that we would like to bound the probability of the ``bad'' event $\CB$ that for at least $m-q$ ellipsoids from $\{\outappr[0.5,n]^i\}_{i=1}^{m-1}$ the bound does not hold, where $\outappr[0.5,n]^i\, \defeq \bigl\{\hspace{0.3mm}\theta \in \mathbb{R}^d : (\theta - \hat{\theta}_n)\tr  {\bar{R}_n} (\theta - \hat{\theta}_n) \leq \gamma_i^*\hspace{0.3mm}\bigl\}$.

Let us fix an index $i\in [m-1],$ and introduce the ``bad'' event $\CB_i$ that the $i$-th ellipsoid is larger than our bound, i.e., \vspace{-1mm}
\begin{equation}
    \CB_i \doteq \left\{ \sup_{\theta_i \in \outappr[0.5,n]^i}
    \hspace{-3mm}
    \|\theta_i - \hat{\theta}_n\| >\!
        \dfrac{f(\delta)(n^{\frac{\rho}{2}}+g^{\frac{1}{2}}(\delta))^{\frac{1}{2}}}{\left(n\lambda_0(n^{\frac{\rho}{2}}-g^{\frac{1}{2}}(\delta))\right)^{\frac{1}{2}}}\right\}.
\end{equation}
From Lemma \ref{lemma:sample_complex_eoa_m2_q1}, we know that $\BP(\CB_i) \leq \delta$, for all $i$. Then, by using a similar argument as in the proof of \cite[Theorem 3]{szentpeteri2023scalar}, we can write the 
overall bad event, i.e., that there are at least $m-q$ ellipsoids for which our bound does not hold, as
\vspace{-1mm}
\begin{equation}
\CB\, = \hspace{-1mm} \bigcap_{\substack{I \subseteq [m-1],\\ |I| \geq m-q}}\hspace{1mm} \bigcup_{i \in I} \;\CB_i \,=\hspace{-1mm} \bigcap_{\substack{I \subseteq [m-1],\\ |I| = m-q}}\hspace{1mm} \bigcup_{i \in I} \;\CB_i.
\end{equation}
By using that $\BP(A \cap B) \leq \min\{\BP(A),\, \BP(B)\}$ and $\BP(A \cup B) \leq \BP(A) + \BP(B)$, 
the probability of $\CB$ can be bounded by
\vspace{-1mm}
\begin{equation}
\BP(\CB)\, \leq \min_{\substack{I \subseteq [m-1],\\ |I| = m-q}}\hspace{0mm} \BP\!\left(\,\bigcup_{i \in I}\,\CB_i \right) \, \leq\, (m-q)\cdot \BP(\CB_1).
\end{equation}
Finally, by introducing $\delta' = (m-q)\delta$, we get
\vspace{-1mm}
    \begin{align}
    \sup_{\theta \in \outappr[p,n]}\|\theta - \hat{\theta}_n\| &\leq
        \dfrac{f(\delta)(n^{\frac{\rho}{2}}+g^{\frac{1}{2}}(\delta))^{\frac{1}{2}}}{\left(n\lambda_0(n^{\frac{\rho}{2}}-g^{\frac{1}{2}}(\delta))\right)^{\frac{1}{2}}}\notag\\
        &=\dfrac{f(\frac{\delta'}{m-q})(n^{\frac{\rho}{2}}+g^{\frac{1}{2}}(\frac{\delta'}{m-q}))^{\frac{1}{2}}}{\left(n\lambda_0(n^{\frac{\rho}{2}}-g^{\frac{1}{2}}(\frac{\delta'}{m-q}))\right)^{\frac{1}{2}}},
    \end{align}
with probability at least $1-\delta',$ which completes the proof.
\end{proof}

\section{Experimental Results}\label{sec:experiments}

{\newtext In this section we illustrate the sizes of the SPS ellipsoidal outer approximation (EOA) confidence regions
via a series of numerical experiments. We consider the problem of robust parameter (coefficient) estimation of an FIR system, i.e., when a confidence ellipsoid is also provided for the given estimate. This problem is a fundamental task in many practical signal processing applications, such as channel estimation \cite{miao2007signal}.

We compare the sizes of the SPS EOA regions with our theoretical bound given in Theorem \ref{thm:sample_complex_eoa}, with the sizes of confidence ellipsoids based on the classical asymptotic \cite{Ljung1999} and novel finite sample PAC theory \cite{Djehiche2021}, furthermore, with set membership ellipsoids \cite{milanese2013bounding}. The asymptotic confidence regions are derived from the central limit theorem. For zero mean independent and identically distributed (i.i.d.)
noise with a finite variance the LSE is asymptotically Gaussian for bounded regressors.
As a consequence, the approximate (asymptotic) confidence regions can be constructed as
\begin{align}\label{equ:asym_conf_region}
    \tilde{\mathcal{C}}_n \defeq \left\{\theta \in \BR^d:(\theta - \hat{\theta}_n)\tr \bar{R}_n (\theta - \hat{\theta}_n) \leq \frac{\mu\hat{\sigma}_n^2}{n}\right\},
\end{align}
where\vspace{-2mm}
\begin{equation}
    \hat{\sigma}^2_n \defeq \frac{1}{n-d} \sum_{i=1}^n \varepsilon_i^2(\hat{\theta}_n),
\end{equation}
and $\varepsilon(\theta)$ is given in Algorithm \ref{alg:sps_indicator}.
The probability that the true parameter $\theta^*$ is in the confidence region $\tilde{\mathcal{C}}_n$ is approximately $F_{\chi^2(d)}(\mu)$, where $F_{\chi^2(d)}$ is the cumulative distribution function of the $\chi^2$ distribution with $d$ degrees of freedom. Then, we have $\mathbb{P}(\theta^* \in \tilde{\mathcal{C}}_n) \approx p$. We applied \cite[Theorem 3.1]{Djehiche2021} to obtain a finite sample PAC bound for the LSE error.
In \cite{Djehiche2021} a different (but equivalent) characterization of a subgaussian variable is used. We will
refer to the following result as the DMR-bound.
\begin{theorem}\label{thm:fir_lse_pac}
    Let $\varphi_1,\dots,\varphi_n$ be the time shifted regressors of a linear
    FIR system with independent centered $\sigma_{\varphi}$-subgaussian entries of unit variance and assume that the noises 
    $\{W_t\}_{t=1}^n$ are $\sigma_{w}$-subgaussian. Then, there is an absolute constant $C>0$ such that, for all $\nu \in (0,1)$, $\eta \in (0,2e^{-1}]$ and as long as
    \begin{align}\label{equ:fir_lse_pac_n_lower}
        n\geq C(\sigma_{\varphi}^2 \vee \sigma_{\varphi}^4)\frac{d \log(d \vee \eta^{-1})}{\nu^2},
    \end{align}
    the LS estimate $\hat{\theta}_n$ satisfies with probability at least $1-\eta$:
    \begin{align}
        \|\hat{\theta}_n - \theta^*\|^2 \leq \frac{2}{(1-\nu)^2}\frac{d}{n}(1+C\sigma_{\varphi}^2\sigma_{w}^2\log^2(2/\eta)).
    \end{align}
\end{theorem}

Our set membership ellipsoid implementation is based on the well-established algorithm 
of \cite{Deller1989}, which recursively builds the confidence ellipsoid around a weighted 
LSE and only updates it when the available data is informative, i.e., 
when a smaller global set membership ellipsoid can be built.
\subsection{Sample complexity for bounded noise}\label{sec:experiment_sc_bounded}

\begin{figure}[!t]
\vspace{1mm}
\centering
\includegraphics[width=3.3in]{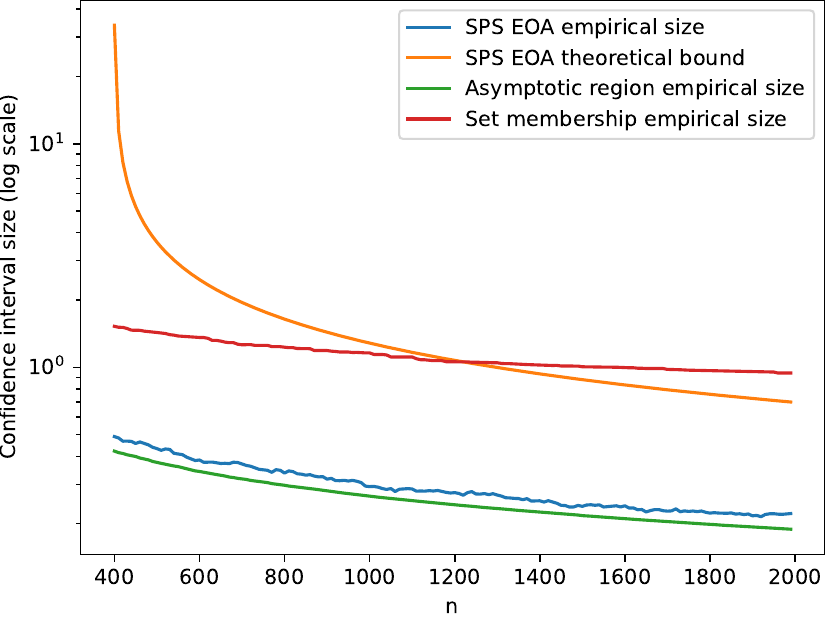}
\caption{{\newtext Size comparison of empirical $0.9$-level SPS ellipsoidal outer approximation, asymptotic and set membership confidence ellipsoids with the theoretical bound of Theorem \ref{thm:sample_complex_eoa} for a uniform noise distribution with $m=10$, $q=1$, $\delta=0.5$, $t_0=400$, $n=2000$ and $s =100$.}}
\label{fig:sa_experiment_unif}
\vspace{-2mm}
\end{figure}

We consider a $2$-dimensional FIR system
\begin{align}\label{equ:fir_exp}
    Y_t = b_1^*U_{t-1} + b_2^*U_{t-2} + W_t,
\end{align}
where $b_1^*=b_2^*=5$. In the first experiment we chose a bounded noise sequence to satisfy the assumption of set membership ellipsoids, hence, $\{W_t\}$ were i.i.d. uniform random variables with zero mean and boundaries $(-3, 3)$. Consequently, the noises were symmetric, nonatomic and $\sigma$-subgaussian with $\sigma=\sqrt{3}$, which is the optimal variance proxy. The input process $\{U_t\}$ was autoregressive, given by
\begin{equation}
    U_t = aU_{t-1} + \sum_{i=1}^5c_iV_{t-i+1},
\end{equation}
with $a=0.7$, $c_1=1$, $c_2=0.775$, $c_3=0.55$, $c_4=0.325$, $c_5=0.1$, and $V_t$ is a
Gaussian white random process with distribution $\mathcal{N}(0, 1)$.
From \eqref{equ:fir_exp}, the linear regression problem can be constructed as $\theta^* = [b_1^*,b_2^*]\tr$ and $\varphi_t = [U_{t-1},U_{t-2}]\tr$.

The SPS EOA confidence regions were generated with hyper-parameters $m = 10$ and $q=1$, i.e., their confidence level was $0.9$.
The parameters that are needed for the sizes of SPS EOA theoretical bounds were computed as follows: $\kappa$ was the largest empirical value of $\frac{t}{d} \max_{1\leq i \leq t}\norm{\Phi_{Q,t}\tr e_i}^2$, while $\lambda_0$ was the smallest empirical eigenvalue of $\bar{R}_{t}$ over all $t_0\leq t \leq n$ and all simulated trajectories. 
Notice that a consequence of setting the value of $\kappa$ this way is that $\rho=1$.
The empirical size of the SPS EOA region was computed as $2\sqrt{\gamma_t^*/\lambda_\text{min}(\bar{R}_{t})}$ for every $t_0\leq t \leq n$, since it corresponds to the longest axis of the ellipsoid $\sup_{\theta \in \outappr[p,n]}2\hspace{0.3mm}\|\theta - \hat{\theta}_n\|$, see \eqref{equ:SPS_EOA_0onf}. The sizes of the asymptotic and set membership confidence ellipsoids were computed accordingly.

We set the sample size to $n = 2000$ and repeated the confidence region constructions for $s = 100$ independently simulated trajectories. The difference between the empirical sizes of SPS EOA, set membership and asymptotic confidence regions, furthermore, the SPS EOA theoretical bounds with confidence level $0.5$ ($\delta = 0.5$) are shown in Fig. \ref{fig:sa_experiment_unif}. To reach the desired confidence level, the median was computed for the empirical sizes in each iteration from $s=100$ trials.

The results of this numerical experiment illustrate that our theoretical bound capture well the empirical decrease rate of the SPS EOA confidence regions. The results also indicate that our theoretical bounds are a bit conservative. This, however, is an expected phenomenon, since these empirical regions are built (a posteriori) with data-driven algorithms, while the theoretical bounds are calculated (a priori) from concentration inequalities. The conservatism of the set membership region can be observed from this experiment, as well, even our theoretical bound has a smaller size for larger sample sizes. Also note that, while the asymptotic region has the smallest size, it lacks the strict theoretical guarantees that the other confidence sets provide, since it is only an approximate region.

\begin{figure}[!t]
\vspace{1mm}
\centering
\includegraphics[width=3.3in]{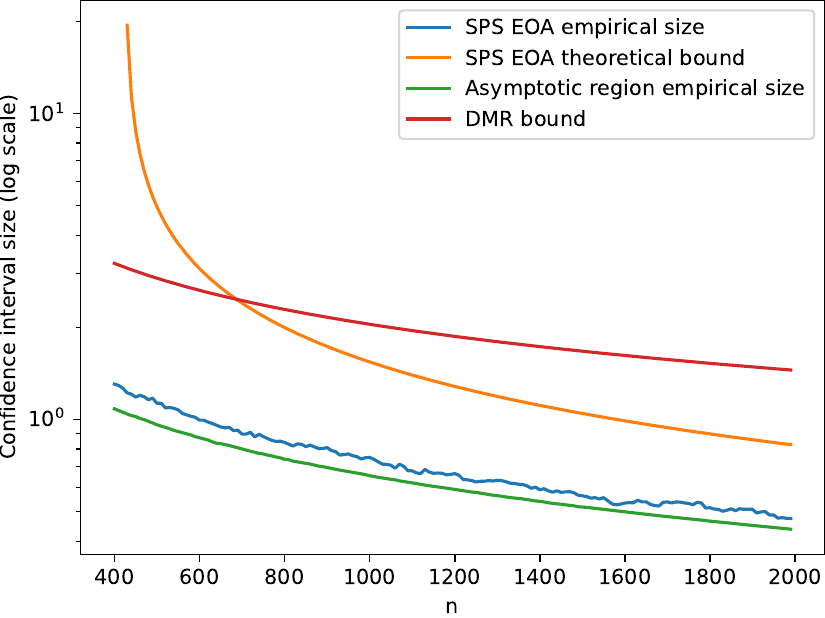}
\caption{{\newtext Comparison of the $0.9$-level empirical and theoretical sizes of SPS EOA regions with asymptotic and DMR bound based ellipsoids for a non-stationary 
Gaussian mixture noise distribution
with $m=10$, $q=1$, $\delta=0.5$, $t_0=400$, $n=2000$ and $s =100$.}}
\label{fig:sa_experiment_nonstatmultimodnorm}
\vspace{-2mm}
\end{figure}

\subsection{Sample complexity for unbounded noise}\label{sec:experiment_sc_unbounded}

In a second experiment we compared our theoretical bound and the empirical sizes of SPS EOA and asymptotic regions with the finite sample DMR bound from \cite[Theorem 3.1]{Djehiche2021} given in Theorem \ref{thm:fir_lse_pac}. We considered a 2-dimensional FIR system with the same structure as in \eqref{equ:fir_exp}, but with a symmetric, non-stationary
Gaussian mixture model
given as
\begin{align}\label{equ:exp_nonstatnoise}
    W_t = \sum_{i=1}^5\mathbb{I}(\zeta=i)Z_{i,t},
\end{align}
where $\zeta \in \{1,\dots,5\}$ has a discrete probability distribution: $\BP(\zeta=1) = 0.3$, $\BP(\zeta=2) = \BP(\zeta=3) = 0.2$, $\BP(\zeta=4) = \BP(\zeta=5) = 0.15$, and $Z_{1,t} \sim \mathcal{N}(0, (n-t)/n + 1)$, $Z_{2,t} \sim \mathcal{N}(-2, 3(n-t)/n)$, $Z_{3,t} \sim \mathcal{N}(2, 3(n-t)/n)$, $Z_{4,t} \sim \mathcal{N}(-5, 2(n-t)/n+1)$, $Z_{5,t} \sim \mathcal{N}(5, 2(n-t)/n+1)$. This type of multimodal noise arises in various signal processing applications where distortions from different sources are modeled. Notice that the noise sequence is independent and symmetric, thus it satisfies the noise assumption on the SPS EOA construction, furthermore, it is $\sigma$-subgaussian with $\sigma=\sqrt{2.7}$. As the distribution is not identically distributed the asymptotic region can be only used as an approximation again. We also changed the input sequence w.r.t. the experiment in section \ref{sec:experiment_sc_bounded} and removed the feedback from it to satisfy the assumption of Theorem \ref{thm:fir_lse_pac}, i.e., the input sequence was $U_t = \sum_{i=1}^5c_iV_{t-i+1}$ with the same parameters as before. From this setting it follows that $\sigma_{\varphi} = (\sum_{i=1}^5c_i^2)^{1/2}$, and we choose the largest $C$ that satisfies \eqref{equ:fir_lse_pac_n_lower} with $\nu = 0.5$. 
Note that \cite{Djehiche2021} does not specify the value of $C$ or how to chose the hyper-parameters of the bound, which makes it challenging to apply the DMR bound in practice.

We again considered $0.9$-level confidence regions, thus we set $m=10$, $q=1$ and $\eta = q/m$. The SPS EOA theoretical bound and the other confidence ellipsoids were computed the same way as before.
Fig. \ref{fig:sa_experiment_nonstatmultimodnorm} illustrates the difference between the empirical sizes of SPS EOA and asymptotic confidence regions with the theoretical DMR and SPS EOA size bounds.

It can be observed that initially (for small sample sizes) our bound on the SPS EOA sizes, given by Theorem \ref{thm:sample_complex_eoa}, is more conservative than the DMR bound of Theorem \ref{thm:fir_lse_pac}. However, our bound decreases much faster as the sample size increases. This also shows that our bound is less conservative than another concentration inequality based uncertainty quantification method when the validity constraint ($n \geq \lceil g^{1/\rho}(\frac{\delta}{m-q}) \rceil$) holds with a reasonable margin. Furthermore, it provides a theoretical guarantee for an algorithm that constructs much smaller confidence ellipsoids (here shown empirically) than the DMR bound.
We 
emphasize again that although the asymptotic region is the smallest, its guarantees are only asymptotic, it can be over-optimistic.

\subsection{Higher-order FIR systems}

{\renewcommand{\arraystretch}{1.4}
\begin{table}[!t]
\newtext
\caption{Comparison of $0.5$-level confidence region sizes for higher-order FIR systems with $m=2$, $q=1$, $\delta=0.5$ and $s =100$.\label{tab:conf_size_compare}}
\centering
\begin{tabular}{cc||cccc}
\hline
d & n & SPS EOA & SPS EOA bound & DMR & Asymptotic\\
\hline
4 & 1000 & 0.79 & 3.67 & 1.93 & 0.76\\
4 & 2000 & 0.55 & 1.47 & 1.92 & 0.53\\
6 & 2000 & 0.73 & 4.82 & 1.70 & 0.69\\
6 & 3500 & 0.53 & 1.61 & 1.69 & 0.51\\
8 & 3500 & 0.61 & 2.80 & 1.59 & 0.62\\
8 & 6000 & 0.45 & 1.33 & 1.58 & 0.44\\
10 & 6000 & 0.53 & 1.81 & 1.50 & 0.51\\
\hline
\end{tabular}
\end{table}}

An experiment was also performed to investigate the sizes of SPS EOA, DMR-bound-based, and asymptotic confidence ellipsoids in higher-order FIR systems. Specifically, we studied
\begin{align}\label{equ:high_dim_fir}
    Y_t = \sum_{k=1}^{d} b_k^*U_{t-k} + W_t.
\end{align}
We considered $\forall \, k : b_k = 5$, the same input sequence $U_t = \sum_{i=1}^5c_iV_{t-i+1}$ and a non-stationary 
Gaussian mixture noise sequence $\{W_t\}$, given in \eqref{equ:exp_nonstatnoise}, as in the previous section. We generated the confidence ellipsoids with confidence level $0.5$, thus we set $m=2$ and $q=1$. The other parameters and the confidence ellipsoid sizes were computed as in Section \ref{sec:experiment_sc_unbounded}. The median of the empirical sizes from $s=100$ independent 
trials and the theoretical bounds 
for $\delta=0.5$ are presented in Table \ref{tab:conf_size_compare} for various sample sizes $n$ and dimensions $d$.

The results show similarity to the two dimensional case. Our theoretical bound provide more conservative sizes than the DMR bound when the sample size $n$ is close to the validity constraint ($n \geq \lceil g^{1/\rho}(\frac{\delta}{m-q}) \rceil$), however, it can be seen that for a larger sample size, given the same dimension $d$, our bound is less conservative, i.e., it decreases faster. Table \ref{tab:conf_size_compare} also shows that the empirical sizes of SPS EOA regions is much smaller than of the DMR bounds and roughly the same as the asymptotic ones.
As a result, 
the order of the FIR system  (its dimension)
does not influence significantly the relative sizes of the regions or the bounds when compared to each other.

\subsection{Computational Complexity}

{
\renewcommand{\arraystretch}{1.4}
\begin{table}[!t]
\newtext
\caption{Comparison of the time required to solve the convex program relative to the $2$-dimensional case for $m=10$.\label{tab:comp_complex}}
\centering
\begin{tabular}{c||cccc}
\hline
\diagbox[]{d}{n} & 250 & 500 & 750 & 1000\\
\hline
\hline
10 & 7.86 & 7.05 & 6.72 & 3.15\\
20 & 32.96 & 28.65 & 24.79 & 11.87\\
30 & 152.04 & 90.44 & 88.21 & 39.54\\
40 & 687.97 & 216.13 & 193.33 & 89.22\\
\hline
\end{tabular}
\end{table}
}

In our last 
experiment the computational complexity of calculating the SPS EOA confidence regions was studied. As one of the main advantages of SPS EOA regions is that they have rigorous guarantees even when there is just a limited number of measurements, we investigated the computational complexity for small sample sizes and relatively high dimensions. The main parameter that influences the computational time of the SPS EOA algorithm is the desired confidence level, since the convex program \eqref{equ:sps_eoA_0vxopt} has to be solved $m$ times. The time required to solve the convex programs in higher-dimensional problems relative to the $2$-dimensional case for different sample sizes and $m=10$ are shown in Table \ref{tab:comp_complex}.

It can be observed that when the dimension of the problem is relatively high compared to the number of observations, the time to solve the resulting SDP grows linearly. Nevertheless, the SPS EOA confidence regions can be computed efficiently on modern hardware, since the two-dimensional variant of the optimization problem can be solved in milliseconds.
Furthermore, the results show that for relatively large $n/d$ ratios the convex problems can be solved considerably faster than for low ratios, showing the efficiency of the construction.
}

{\newtext
\section{Discussion}
Confidence regions, typically based on concentration inequalities, can provide high probability bounds on the estimation errors and  have become one of the main tools for uncertainty quantification over the past decades. This field of study is important both in statistical learning and in engineering applications, where robustness is critical, such as adaptive control and filtering. Here, we focused on robust estimation of exogenous linear systems, including FIR models, which play a vital role in signal processing across various domains, e.g., communications, biomedical engineering, and radar systems.

One of the key advancement of this work is that we rigorously prove concentration inequality based high probability upper bounds on the sizes of SPS confidence ellipsoids, which are based on the solutions of convex semidefinite programs. The main advantages of SPS compared to the aforementioned methods is that it can build confidence regions without any specific distribution or moment assumptions in a data-driven fashion, i.e., without the need of any hyper-parameters. Our analysis provides rigorous guarantees for the sizes of these SPS ellipsoids, which can be compared to the result of \cite{Djehiche2021}. Although it considers centered random regressor sequences, while our bound assumes deterministic regressors, the two bounds have the same decrease rate as the sample size increases and their dependence on the dimension and noise parameter is also similar. However, the DMR bound depends on hyper-parameters that are typically unknown in practice. Moreover, as we showed in our experiment for a robust FIR parameter estimation problem and a Gaussian 
input process, our bound gets less conservative than the DMR bound of \cite{Djehiche2021} as the sample size increases. Furthermore, the empirical sizes of SPS EOA regions are much smaller than the DMR bound.

Therefore, as our 
finite sample PAC bound and the experimental evaluations demonstrate, SPS EOA has several advantages compared to concentration inequality based uncertainty quantification methods, such as a data-driven construction, smaller theoretical and significantly smaller empirical sizes. 

One of the main assumptions behind SPS is that the noise terms are independent and symmetric about zero. A trade-off regarding this assumption can be made by using the Residual-Permuted Sums (RPS) algorithm \cite{szentpeteri2024rps}, which builds confidence regions assuming that the noise terms are independent and identically distributed. A future research direction could be to analyze the sample complexity of RPS based confidence ellipsoids and compare them to the bounds of Theorems \ref{thm:sample_complex_eoa} and \ref{thm:fir_lse_pac}. 
An instrumental variable based extension of SPS to state-space models can be found in \cite{szentpeteri2023}, while \cite{Care2025} proposed a generalization for ARX models.
Another possible research direction could be to investigate the sizes of SPS confidence ellipsoids in case of closed-loop state space models or ARX systems and compare them with the results of \cite{simchowitz18a, Jedra2023}. 
}
\section{Conclusion}
In this paper we have analyzed the sample complexity of 
data-driven confidence ellipsoids for linear regression problems which are constructed as outer approximations of the Sign-Perturbed Sums (SPS) confidence regions. These confidence ellipsoids share their centers (i.e., the least-squares estimate) and shape matrices (i.e., the empirical covariance of the regressors) with the classical asymptotic ellipsoids, only their radii are different. These radii can be calculated by convex programming methods resulting in distribution-free confidence ellipsoids with finite sample coverage guarantees.

Our results build upon the theory of concentration inequalities and give high probability upper bounds on the sizes of these confidence ellipsoids under the assumptions that the observation noises are independent, symmetric, nonatomic and subgaussian, as well as that the regressors are suitably exciting. We have also showed that the sizes shrink at the optimal rate.

Future research directions include extending our results to dynamical systems, e.g., to ARX and state-space models.

{\appendices

\section{Proof of Lemma \ref{lemma:sample_complex_eoa_m2_q1}}\label{sec:proof_sample_complex_eoa_m2_q1}
\begin{proof}
The results of \eqref{equ:eoa_concentration_firstpart} and \eqref{equ:eoa_concentration_secpart} will be combined to provide the claimed stochastic lower bound.

From \eqref{equ:eoa_concentration_firstpart}, if $0 \leq \varepsilon \leq 8\sigma^2 d^2$, we have
\begin{align}
    \BP\left(\frac{\vert X_{\scriptscriptstyle P} - \mathbb{E}X_{\scriptscriptstyle P}\vert}{n\lambda_0} \geq \varepsilon\right) \leq 2\exp\left(-\frac{\varepsilon^2n^2\lambda_0^2}{64d^2\sigma^4}\right)\!,
\end{align}
which can be reformulated by introducing $\delta \defeq 4\exp(-\varepsilon^2n^2\lambda_0^2/(64d^2\sigma^4))$ as, for all $\delta: 4\exp(-(nd\lambda_0)^2) \leq \delta \leq 1$, with probability (w.p.) at least $1-\delta/2$, it holds that
\begin{align}
    \frac{\vert X_{\scriptscriptstyle P} - \mathbb{E}X_{\scriptscriptstyle P}\vert}{n\lambda_0} \leq \frac{8d\sigma^2\ln^{\frac{1}{2}}(\tfrac{4}{\delta})}{n\lambda_0}.
\end{align}
Likewise, if $\varepsilon > 8\sigma^2{d^2}$, for all $\delta$, such that $0 \leq \delta < 4\exp(-(nd\lambda_0)^2) $ we have, w.p. at least $1-\delta/2$, that
\begin{align}
    \frac{\vert X_{\scriptscriptstyle P} - \mathbb{E}X_{\scriptscriptstyle P}\vert}{n\lambda_0} \leq \frac{8\sigma^2\ln(\tfrac{4}{\delta})}{n\lambda_0}.
\end{align}
Combining these together we get, w.p.\ at least $1-\delta/2$, that 
\begin{align}
    &\frac{\vert X_{\scriptscriptstyle P} - \mathbb{E}X_{\scriptscriptstyle P}\vert}{n\lambda_0} \leq \notag\\
    &\begin{cases}
        \dfrac{8d\sigma^2\ln^{\frac{1}{2}}(\tfrac{4}{\delta})}{n\lambda_0} \vspace{1mm}& 4e^{-(nd\lambda_0)^2} \leq \delta \leq 1,\\
        \dfrac{8\sigma^2\ln(\tfrac{4}{\delta})}{n\lambda_0} & 0 \leq \delta < 4e^{-(nd\lambda_0)^2}. 
    \end{cases}
\end{align}
In the proof of Lemma \ref{lemma:numerator_indicator} \cite[Appendix B]{szentpeteri2025} it is shown that a random variable in the form of $X = w_n\tr  M w_n$, where $M$ is a projection matrix with rank($M$)$=d$, can be upper bounded as $X \leq \sum_{i=1}^d \tilde{w}_{d,i}^2$, where $\{\tilde{w}_{d,i}\}$ are (zero mean) $\sigma$-subgaussians. Since $X_{\scriptscriptstyle P}$ is in this form, we have that
\vspace{-1mm}
\begin{align}
    \BE\big[ X_{\scriptscriptstyle P}\big] &\leq \BE\left[\sum_{t=1}^d \tilde{w}_{d,i}^2\right] = \sum_{t=1}^d\BE\left[\tilde{w}_{d,i}^2\right] =
    \sum_{t=1}^d \text{Var}\left[\tilde{w}_{d,i}\right] \notag\\
    &\leq \sum_{t=1}^d\sigma^2 = d\,\sigma^2.
\end{align}
Applying the reverse triangle inequality and that $\mathbb{E}X_{\scriptscriptstyle P} \leq d\sigma^2$:
\begin{align}
    \frac{\vert X_{\scriptscriptstyle P} - \mathbb{E}X_{\scriptscriptstyle P}\vert}{n\lambda_0} \geq \frac{\vert X_{\scriptscriptstyle P}\vert - \vert\mathbb{E}X_{\scriptscriptstyle P}\vert}{n\lambda_0} \geq \frac{\vert X_{\scriptscriptstyle P}\vert}{n\lambda_0} - \frac{d\sigma^2}{n\lambda_0},
\end{align}
and we have w.p. at least $1-\delta/2$ that
\vspace{-1mm}
\begin{align}\label{equ:num_result}
    &\frac{\vert X_{\scriptscriptstyle P}\vert}{n\lambda_0} \leq \notag\\
    &\begin{cases}
        \dfrac{d\sigma^2\left(8\ln^{\frac{1}{2}}(\tfrac{4}{\delta})+1\right)}{n\lambda_0}\vspace{1mm} & 4e^{-(nd\lambda_0)^2} \leq \delta \leq 1,\\
        \dfrac{\sigma^2\left(8\ln(\tfrac{4}{\delta})+d\right)}{n\lambda_0} & 0 \leq \delta < 4e^{-(nd\lambda_0)^2}.
    \end{cases}
\end{align}
Next, the concentration inequality result of \eqref{equ:eoa_concentration_secpart} is reformulated as, w.p. at least $1-\delta/2$, we have
\vspace{-2mm}
\begin{align}\label{equ:den_result}
    &\frac{1+\lambda^{1/2}_{\text{max}}(K^2)}{1-\lambda^{1/2}_{\text{max}}(K^2)}\leq \frac{1+\left(\frac{2\kappa d^2\ln\left(\frac{4d}{\delta}\right)}{n^{\rho}}\right)^{\frac{1}{2}}}{1-\left(\frac{2\kappa d^2\ln\left(\frac{4d}{\delta}\right)}{n^{\rho}}\right)^{\frac{1}{2}}}\notag\\
    &=\frac{n^{\frac{\rho}{2}}+\left(2\kappa d^2\ln\left(\frac{4d}{\delta}\right)\right)^{\frac{1}{2}}}{n^{\frac{\rho}{2}}-\left(2\kappa d^2\ln\left(\frac{4d}{\delta}\right)\right)^{\frac{1}{2}}} = \frac{n^{\frac{\rho}{2}}+g^{\frac{1}{2}}(\delta)}{n^{\frac{\rho}{2}}-g^{\frac{1}{2}}(\delta)},
\end{align}
where we used the definition of $g(\delta)$ \eqref{equ:f_and_g_not_appendix}. Using the union bound it can be shown that if
\begin{align}\label{equ:union_1}
    &\BP(Y_1\leq y_1) \geq 1-p_1, &\BP(Y_2\leq y_2) \geq 1-p_2,
\end{align}
then\vspace{-2mm}
\begin{align}\label{equ:union_2}
    \BP(Y_1Y_2\leq y_1y_2) \geq 1-(p_1 + p_2).
\end{align}
Combining the result of \eqref{equ:union_1}-\eqref{equ:union_2} with the stochastic lower bounds of \eqref{equ:num_result} and \eqref{equ:den_result}
we conclude that w.p. at least $1-\delta$
    \begin{align}
        &\frac{\vert X_{\scriptscriptstyle P}\vert\left(1+\lambda^{1/2}_{\text{max}}(K^2)\right)}{n\lambda_0\left(1-\lambda^{1/2}_{\text{max}}(K^2)\right)} \leq \notag\\
        &\begin{cases}
            \dfrac{d\sigma^2\left(8\ln^{\frac{1}{2}}(\tfrac{4}{\delta})+1\right)\left(n^{\frac{\rho}{2}}+g^{\frac{1}{2}}(\delta)\right)}{n\lambda_0\left(n^{\frac{\rho}{2}}-g^{\frac{1}{2}}(\delta)\right)}\vspace{1mm} & 4e^{-(nd\lambda_0)^2} \leq \delta \leq 1,\\
            \dfrac{\sigma^2\left(8\ln(\tfrac{4}{\delta})+d\right)\left(n^{\frac{\rho}{2}}+g^{\frac{1}{2}}(\delta)\right)}{n\lambda_0\left(n^{\frac{\rho}{2}}-g^{\frac{1}{2}}(\delta)\right)} & 0\leq \delta < 4e^{-(nd\lambda_0)^2}. \\
        \end{cases}\notag\\[-5mm]
    \end{align}
    The above stochastic lower bound can be applied to obtain a high probability upper bound for the size of the 0.5-level SPS EOA region. It is shown in \eqref{equ:eoa_ellipsoid_upperbound}, that for every $\theta \in \outappr[0.5,n]$:\vspace{-1mm}
    \begin{align}
        \|\theta - \hat{\theta}_n\|^2 \leq \frac{\tfrac{1}{n}X_{\scriptscriptstyle P}(1+\lambda^{1/2}_{\text{max}}(K^2))}{\lambda_0(1-\lambda^{1/2}_{\text{max}}(K^2))},
    \end{align}
    consequently, it holds w.p.\ at least $1-\delta$ that
\begin{align}
    &\sup_{\theta \in \outappr[0.5,n]}\|\theta - \hat{\theta}_n\| \leq 
        \dfrac{f(\delta)(n^{\frac{\rho}{2}}+g^{\frac{1}{2}}(\delta))^{\frac{1}{2}}}{\left(n\lambda_0(n^{\frac{\rho}{2}}-g^{\frac{1}{2}}(\delta))\right)^{\frac{1}{2}}},
\end{align}
where we used the definition of $f(\delta)$ from \eqref{equ:f_and_g_not_appendix}.
\end{proof}

\section{Technical Lemmas}
\begin{lemma}\label{lemma:S1_LS_notnull}
Assuming \ref{assu:noise} and that the SPS confidence region is bounded, it holds almost surely that
\begin{align}
    \|S_1(\hat{\theta}_n)\|^2 \neq 0.
\end{align}
\end{lemma}
\begin{proof}
Using our notation \eqref{equ:def_D}-\eqref{equ:def_Q} $S_1(\theta)$ can be written as
\begin{align}
    S_1(\theta) = \bar{R}_n^{-\frac{1}{2}}\frac{1}{n}\left(Q_n(\theta^* - \theta) + \Phi_n\tr  D_{\alpha, n}w_n\right),
\end{align}
therefore
\begin{align}
    S_1(\hat{\theta}_n) &= \frac{1}{n}\bar{R}_n^{-\frac{1}{2}}\left(Q_n(\theta^* - R_n^{-1}\Phi_n\tr  y_n) + \Phi_n\tr  D_{\alpha, n}w_n\right)\notag\\
    &= \frac{1}{n}\bar{R}_n^{-\frac{1}{2}}\left(Q_n(- R_n^{-1}\Phi_n\tr w_n) + \Phi_n\tr  D_{\alpha, n}w_n\right)\notag\\
    &=\frac{1}{n}\bar{R}_n^{-\frac{1}{2}}\Phi_n\tr D_{\alpha, n}\left(I- \Phi_nR_n^{-1}\Phi_n\tr\right)w_n.
\end{align}
In \cite{Care2022} it is shown that if the confidence region is bounded,
then $A$ is positive definite, consequently $\bar{R}_n^{-\frac{1}{2}}A\bar{R}_n^{-\frac{1}{2}}$ is also positive definite (\ref{assu:R_nonsigular}), therefore it is full rank. Using our definition of $A$ \eqref{equ:def_A}, it can be written that $\bar{R}_n^{-\frac{1}{2}}A\bar{R}_n^{-\frac{1}{2}} = \bar{R}_n^{-\frac{1}{2}}\Phi\tr D_{\alpha,n}(I-\Phi_nR_n^{-1}\Phi_n\tr)D_{\alpha,n}\Phi\bar{R}_n^{-\frac{1}{2}}$. Lets denote $L \defeq \bar{R}_n^{-\frac{1}{2}}\Phi\tr D_{\alpha,n}(I-\Phi_nR_n^{-1}\Phi_n\tr)$. Since $L$ is a part of the product of the $d$-ranked matrix $\bar{R}_n^{-\frac{1}{2}}A\bar{R}_n^{-\frac{1}{2}}$, it follows that rank$(L) = d$. Notice that $S_1(\hat{\theta}_n) = 1/n \cdot L w_n$, therefore by applying Lemma \ref{lemma:rank_1_dot_nonatomic} we conclude that $S_1(\hat{\theta}_n) \neq 0$ almost surely, hence $\|S_1(\hat{\theta}_n)\|^2 \neq 0$ almost surely.
\end{proof}
\begin{lemma}\label{lemma:condition2_not_null}
    Assuming \ref{assu:noise}, \ref{assu:R_nonsigular} and $\lambda_{\text{max}}(K^2) < 1$, the following holds almost surely for $\xi = 1/(1-\lambda_{\text{max}}(K^2))$,
    \begin{align}\label{equ:psd_cond_2_notnull}
        \left(I-A'_0(\xi)\left(A'_0(\xi)\right)^\dagger\right)\xi b_0 \neq 0,
    \end{align}
    where $A'_0(\xi)$ and $b_0$ are defined in \eqref{equ:def_Ac-coma} and \eqref{equ:def_bc}, respectively.
\end{lemma}
\begin{proof}
    Using the reformulation of $\xi A'_0(\xi)$ from \eqref{equ:A_psi_reform}, we have
    \begin{align}\label{equ:eigen_Ac'}
        A'_0(\xi) & =  V_{\Phi_{R}}U_{\Phi_{R}}\tr  (-\tfrac{1}{\xi}I + I-K^2) U_{\Phi_{R}}V_{\Phi_{R}}\tr \notag\\
        & = V_{\Phi_{R}}U_{\Phi_{R}}\tr  V_K (-\tfrac{1}{\xi}I + I - \Lambda_K^2) V_K\tr  U_{\Phi_{R}}V_{\Phi_{R}}\tr,
    \end{align}
    where $K = V_K \Lambda_K V_K\tr $ is
    the eigendecomposition of $K$.
    Substituting $\xi = 1/(1-\lambda_{\text{max}}(K^2))$ to the first part of \eqref{equ:psd_cond_2_notnull}
    and using the decomposition above, we get
    \begin{align}\label{equ:def_D1}
        &\left(I- A'_0\left(\tfrac{1}{1-\lambda_{\text{max}}(K^2)}\right) \left(A'_0\left(\tfrac{1}{1-\lambda_{\text{max}}(K^2)}\right)\right)^\dagger\right) = \notag\\
        &V_{\Phi_{R}}U_{\Phi_{R}}\tr V_KD_{\text{max}} V_K\tr U_{\Phi_{R}}V_{\Phi_{R}}\tr,
    \end{align}
    with $D_{\text{max}} \defeq \text{diag}(0,\dots,0,1,\dots,1) \in \BR^{d \times d}$, where the number of ones equals to the multiplicity of $\lambda_{\text{max}}(K^2)$ in $\Lambda_K^2$.
    
    The vector $b_0$ defined in \eqref{equ:def_bc} can be rewritten by using the reformulation from \eqref{equ:y_to_w_rewrite} as
    \begin{align}
        b_0 & = \frac{1}{n^{\frac{1}{2}}}R_n^{-\frac{1}{2}}Q_nR_n^{-1}B D_{\alpha,n} y_n\notag\\
            & = \frac{1}{n^{\frac{1}{2}}}R_n^{-\frac{1}{2}}Q_nR_n^{-1}B D_{\alpha,n} w_n.
    \end{align}
    Let
    \begin{align}
        B_0 \defeq R_n^{-\frac{1}{2}}Q_nR_n^{-1}B D_{\alpha,n}.
    \end{align}
    Notice that $BB\tr  = A$, since
    \begin{align}\label{equ:BBtr}
    BB\tr &= \left(\Phi_n\tr  - Q_nR_n^{-1}\Phi_n\tr  D_{\alpha,n}\right)\left(\Phi_n - D_{\alpha,n}\Phi_n R_n^{-1}Q_n\right) \notag\\
    &=\left(R_n - Q_nR_n^{-1}Q_n - Q_nR_n^{-1}Q_n+ Q_nR_n^{-1}Q_n\right)\notag\\
    &=A.
    \end{align}
    Then, it holds that
    \begin{align}\label{equ:B_0B_0_T}
        B_0B_0\tr & = R_n^{-\frac{1}{2}}Q_nR_n^{-1}B D_{\alpha,n}D_{\alpha,n} B\tr  R_n^{-1}Q_nR_n^{-\frac{1}{2}} \notag\\
                  & = R_n^{-\frac{1}{2}}Q_nR_n^{-1}AR_n^{-1}Q_nR_n^{-\frac{1}{2}}.
    \end{align}
    Using the same reformulations as in \eqref{equ:reformulation_A}-\eqref{equ:reformulation_A_2}, specifically $A = \Phi_{R,n}\tr (I - K^2)\Phi_{R,n}$, $R^{-1} = \Phi_{R,n}^{-1}\Phi_{R,n}\tri$ and $Q_n=\Phi_{R,n}\tr K\Phi_{R,n}$ it can be written that 
    \begin{align}\label{equ:BO_middle_term}
        Q_nR_n^{-1}AR_n^{-1}Q_n = \Phi_{R,n}\tr K(I - K^2)K\Phi_{R,n}.
    \end{align}
    Writing back \eqref{equ:BO_middle_term} to \eqref{equ:B_0B_0_T} and applying 
    $R_n^{-1/2}=V_{\Phi_{R}}\Sigma_{\Phi_{R}}^{-1}V_{\Phi_{R}}\tr $ and $\Phi_{R,n} = U_{\Phi_{R}}\Sigma_{\Phi_{R}}V_{\Phi_{R}}\tr$ from \eqref{equ:Phi_R_SVD}, we get
    \begin{align}
        B_0B_0\tr = V_{\Phi_{R}}U_{\Phi_{R}}\tr K(I - K^2)KU_{\Phi_{R}}V_{\Phi_{R}}\tr.
    \end{align}
    Finally, using again the eigendecomposition of $K$
    \begin{align}\label{equ:B_0B_0_T_final}
        B_0B_0\tr = V_{\Phi_{R}}U_{\Phi_{R}}\tr V_K \Lambda_K(I - \Lambda_K^2)\Lambda_K V_K\tr U_{\Phi_{R}}V_{\Phi_{R}}\tr.
    \end{align}

    Notice the similarity between \eqref{equ:eigen_Ac'} and \eqref{equ:B_0B_0_T_final}. In these formulas $-\tfrac{1}{\xi}I + I - \Lambda_K^2$ and $\Lambda_K(I - \Lambda_K^2)\Lambda_K$ are diagonal and $V_{\Phi_{R}}U_{\Phi_{R}}\tr V_K$ is orthonormal, since $V_K\tr U_{\Phi_{R}}  V_{\Phi_{R}}\tr V_{\Phi_{R}}U_{\Phi_{R}}\tr V_K = I$. It follows that  \eqref{equ:eigen_Ac'} and \eqref{equ:B_0B_0_T_final} are eigendecompositions of $A'_0(\xi)$ and $B_0B_0\tr$, therefore they can be written as $A'_0(\xi) = V_{A_{0}}(-\tfrac{1}{\xi}I + I - \Lambda_K^2) V_{A_{0}}\tr$ and $B_0B_0\tr =  V_{A_{0}}\Lambda_K(I - \Lambda_K^2)\Lambda_K V_{A_{0}}\tr$, where we used the same arrangements of eigenvalues and eigenvectors.
    Using the SVD decomposition of $B_0 = U_{B_{0}}\Sigma_{B_{0}}V_{B_{0}}\tr $ and the above eigendecomposition $B_0B_0\tr =  V_{A_{0}}\Lambda_K(I - \Lambda_K^2)\Lambda_K V_{A_{0}}\tr$ it holds that 
    \begin{align}\label{equ:B_0_svd_equ_A0_eig}
        U_{B_{0}}\Sigma_{B_{0}}^2U_{B_{0}}\tr  = V_{A_{0}} \Lambda_K(I-\Lambda_K^2) \Lambda_K V_{A_{0}}\tr .
    \end{align}
    From equation \eqref{equ:B_0_svd_equ_A0_eig}, it follows that there is an arrangement of eigenvalues in $\Sigma_{B_{0}}^2$ and $\Lambda_K(I-\Lambda_K^2) \Lambda_K$ that for their corresponding eigenvectors $U_{B_{0}} = V_{A_{0}}$. Using this arrangement of eigenvectors and eigenvalues, i.e., $U_{B_{0}} = V_{A_{0}}$ and $\Sigma_{B_{0}} = ((I - \Lambda_K^2)\Lambda_K^2)^{1/2}$, furthermore our result from \eqref{equ:def_D1}, we conclude that in case $\xi = 1/(1-\lambda_{\text{max}}(K^2))$, we have
    \begin{align}
        &\left(I- A'_0(\xi)\left(A'_0(\xi)\right)^\dagger\right)\xi b_0 = \notag\\
        &\frac{1}{1-\lambda_{\text{max}}(K^2)}V_{A_{0}}D_{\text{max}}V_{A_{0}}\tr  B_0  w_n = \notag\\
        &\frac{1}{n^{\frac{1}{2}}(1-\lambda_{\text{max}}(K^2))}V_{A_{0}}D_{\text{max}}V_{A_{0}}\tr  U_{B_{0}}\Sigma_{B_{0}}V_{B_{0}}\tr  w_n = \notag\\
        &\frac{1}{n^{\frac{1}{2}}(1-\lambda_{\text{max}}(K^2))}U_{B_{0}}D_{\text{max}}((I - \Lambda_K^2)\Lambda_K^2)^{\frac{1}{2}}V_{B_{0}}\tr w_n=\notag\\
        &\frac{\max_i\abs{\lambda_i(K)}}{(n(1-\lambda_{\text{max}}(K^2)))^{\frac{1}{2}}} U_{B_{0}}D_{\text{max}}V_{B_{0}}\tr  w_n,
    \end{align}
    where in the last step we used that $D_{\text{max}}$ ``selects'' the largest eigenvalues in the ordering, as in \eqref{equ:def_D1}, which equals $\max_i\abs{\lambda_i(K)}\cdot(1-\lambda_{\text{max}}(K^2))^{1/2}$. It was shown in the proof of Theorem \ref{thm:sample_complex_eoa} that the eigenvalues of matrix $K$ satisfy $\forall i: 0 < \abs{\lambda_i(K)} \leq 1$ and we assume that $\lambda_{\text{max}}(K^2) < 1$, therefore $(\max\abs{\lambda_i(K)})/((n(1-\lambda_{\text{max}}(K^2)))^{1/2}) \neq 0$. It holds that $1 \leq \text{rank}(U_{B_{0}}D_{\text{max}}V_{B_{0}}\tr) \leq d$, which follows from the fact that the number of ones in $D_{\text{max}}$ equals to the multiplicity of $\lambda_{\text{max}}(K^2)$ in $\Lambda_K^2$. Then, $U_{B_{0}}D_{\text{max}}V_{B_{0}}\tr  w_n$ is almost surely nonzero, as $w_n$ consist of independent, nonatomic random variables (\ref{assu:noise}), therefore Lemma \ref{lemma:rank_1_dot_nonatomic} can be applied.
\end{proof}

\begin{lemma}\label{lemma:rank_1_dot_nonatomic}
    Let $M \in \BR^{k \times n}$ 
    with $\text{rank}(M)>0$
    and  $w_0 \in \BR^{n}$ a vector of independent nonatomic random variables. Then, 
    \begin{align}
        \mathbb{P}( M w_0 \neq 0 ) = 1.
    \end{align}
\end{lemma}
\begin{proof}
Since $\text{rank}(M)>0$, there is at least one row of $M$ which is nonzero. Let $\Vec{m}$ denote any of these row vectors. Observe that if it enough to prove that $\Vec{m} w_0 \neq 0$ almost surely.
There exists an index $i$ with $\Vec{m}_{i} \neq 0$. Then, we have
\begin{align}
    \Vec{m}w_0 &= \sum_{t=1}^n\Vec{m}_{t}W_{0,t} = \Vec{m}_{i}\left(W_{0,i} + \sum_{\substack{t=1, t \neq i}}^n\frac{\Vec{m}_{t}}{\Vec{m}_{i}}W_{0,t}\right)\notag\\
    &= \Vec{m}_{i}\left(W_{0,i} + \sum_{\substack{t=1, t \neq i}}^n\Vec{m}'_{t}W_{0,t}\right),
\end{align}
where $w_0=[W_{0,1}, W_{0,2}, \dots, W_{0,n}]\tr$. $Mw_0 = 0$, if $W_{0,i} + \sum_{\substack{t=1, t \neq i}}^n\Vec{m}'_{t}W_{0,t}=0$, hence we investigate the probability $\BP\left(W_{0,i} + \sum_{\substack{t=1, t \neq i}}^n\Vec{m}'_{t}W_{0,t}=0\right)$. Using the law of total expectation it can be derived that
    \begin{align}
        &\BP\left(W_{0,i} + \sum_{\substack{t=1, t \neq i}}^n\Vec{m}'_{t}W_{0,t}=0\right) = \\
        &\BE\left[\BP\left(W_{0,i} + \sum_{\substack{t=1, t \neq i}}^n\Vec{m}'_{t}W_{0,t}=0\Bigg\vert \{W_{0,t}\}_{\substack{t=1,t \neq i}}^n\right)\right] = 0,\notag
    \end{align}
    since $W_{0,i}$ is nonatomic for every $i$.
\end{proof}

\begin{lemma}\label{lemma:eig_P}
    Let $P(\xi)$ be defined as in \eqref{equ:P_defined}, that is
    \begin{align}\label{equ:P_defined_general}
    P(\xi)\,=\; &D_{\alpha,n}B\tr \Big[R_n^{-1}Q_n R_n^{-\frac{1}{2}}\left(A'_0(\xi)\right)^\dagger R_n^{-\frac{1}{2}}Q_n R_n^{-1}\notag\\
    &+R_n^{-1}\Big]BD_{\alpha,n}.
    \end{align}
    Assume \ref{assu:R_nonsigular}, $\forall\,i: 0<\abs{\lambda_i(K)} < 1$ and $\xi >1/(1-\lambda_{\text{max}}(K^2))$. Then, $P(\xi)$ is positive semidefinite and $\text{rank}(P(\xi)) = d$, furthermore for every 
    nonzero eigenvalue of $P(\xi)$, 
    \begin{align}
        \lambda_i(P(\xi)) = \frac{\xi\lambda_i^2(K) - \xi - \lambda_i^2(K) + 1}{\xi\lambda_i^2(K) - \xi + 1},
    \end{align}
    where $K$ is defined in \eqref{equ:def_K}. Note that $\lambda_i(\cdot)$ are ordered, that is $\lambda_1(\cdot) \geq \lambda_2(\cdot) \geq \cdots \geq \lambda_d(\cdot)$.
\end{lemma}
\begin{proof}
    Lets denote the middle part of the product given in the definition of $P(\xi)$ \eqref{equ:P_defined_general} as
    \begin{align}
    P_{1}(\xi) \defeq R_n^{-1}Q_nR_n^{-\frac{1}{2}}\left(A'_0(\xi)\right)^\dagger R_n^{-\frac{1}{2}}Q_nR_n^{-1} + R_n^{-1}.     
    \end{align}
    Using the same reformulations as in \eqref{equ:reformulation_A}-\eqref{equ:lam_lb_end}, namely
    $R^{-1} = \Phi_{R,n}^{-1}\Phi_{R,n}\tri$, $Q_n=\Phi_{R,n}\tr K\Phi_{R,n}$, $R_n^{-1/2}=V_{\Phi_{R}}\Sigma_{\Phi_{R}}^{-1}V_{\Phi_{R}}\tr $ and $\Phi_{R,n} = U_{\Phi_{R}}\Sigma_{\Phi_{R}}V_{\Phi_{R}}\tr$, we have
    \begin{align}
        P_{1}(\xi) &= \Phi_{R,n}^{-1}\Big(K\Phi_{R,n}R_n^{-\frac{1}{2}}\left(A'_0(\xi)\right)^\dagger R_n^{-\frac{1}{2}}\Phi_{R,n}K+ I\Big)\Phi_{R,n}\ntr \notag \\
        &= \Phi_{R,n}^{-1}\Big(KU_{\Phi_{R}}V_{\Phi_{R}}\tr\left(A'_0(\xi)\right)^\dagger V_{\Phi_{R}}U_{\Phi_{R}}\tr K+ I\Big)\Phi_{R,n}\ntr .
    \end{align}
    The matrix $A'_0(\xi)$ can be written as in \eqref{equ:A_psi_reform}
    \begin{align}
        A'_0(\xi)  =  V_{\Phi_{R}}U_{\Phi_{R}}\tr  (-\tfrac{1}{\xi}I + I-K^2) U_{\Phi_{R}}V_{\Phi_{R}}\tr,
    \end{align}
    therefore, 
    \begin{align}
        P_{1}(\xi) = \Phi_{R,n}^{-1}\left(K\left(-\tfrac{1}{\xi}I + I-K^2\right)^\dagger K + I\right)\Phi_{R,n}\ntr.
    \end{align}
    Using the eigendecomposition of $K = V_K\Lambda_KV_K\tr$ it holds that
    \begin{align}\label{equ:P1_reform}
        P_{1}(\xi) &= \Phi_{R,n}^{-1}V_K\left(\Lambda_K\left(I-\Lambda_K^2-\tfrac{1}{\xi}I\right)^\dagger\Lambda_K + I\right)V_K\tr\Phi_{R,n}\ntr\notag\\
        & = \Phi_{R,n}^{-1}V_K\Lambda'_K(\xi)V_K\tr\Phi_{R,n}\ntr,
    \end{align}
    where
    \begin{align}\label{equ:def_Lam_K_coma}
        \Lambda'_K(\xi) \defeq \Lambda_K\left(I-\Lambda_K^2-\tfrac{1}{\xi}I\right)^\dagger\Lambda_K + I.
    \end{align}
    Notice that as we assume $\xi >1/(1-\lambda_{\text{max}}(K^2))$, it holds for every diagonal element (eigenvalue) of $(I-\Lambda_K^2-{I}/{\xi})$ that
    \begin{align}
        1-\lambda_{i}(K^2)-\frac{1}{\xi} &> 1-\lambda_{i}(K^2)-(1-\lambda_{\text{max}}(K^2))\notag\\
        &=\lambda_{\text{max}}(K^2)-\lambda_{i}(K^2) \geq 0,
    \end{align}
    therefore $(I-\Lambda_K^2-{I}/{\xi})$ is positive definite. Since we assume that $\forall\,i: 0<\abs{\lambda_i(K)} < 1$, it follows that $\Lambda'_K(\xi)$ is also positive definite.
    Substituting our reformulation of $P_{1}(\xi)$ from \eqref{equ:P1_reform} back to $P(\xi)$ \eqref{equ:P_defined_general}, we get
    \begin{align}\label{equ:P_reform}
        &P(\xi)=D_{\alpha,n}B\tr \Phi_{R,n}^{-1}V_K\Lambda'_K(\xi)
        V_K\tr\Phi_{R,n}\ntr BD_{\alpha,n}.
    \end{align}
    
    By introducing $E \defeq \Phi_{R,n}\ntr BD_{\alpha,n}$, recalling that $BB\tr = A$ \eqref{equ:BBtr} and using the reformulation of $A$ from \eqref{equ:reformulation_A}-\eqref{equ:reformulation_A_2} with the eigendecomposition of $K$, it can be derived that
    \begin{align}\label{equ:B_tilde}
        EE\tr &= \Phi_{R,n}\ntr BD_{\alpha,n}D_{\alpha,n}B\tr\Phi_{R,n}^{-1} = \Phi_{R,n}\ntr A\Phi_{R,n}^{-1} \notag\\
        &=I-K^2 = V_K(I-\Lambda_K^2)V_K\tr.
    \end{align}
    It follows that the full SVD decomposition of $E$ is
    \begin{align}\label{equ:B_tilde_svd}
        E = \Phi_{R,n}\ntr BD_{\alpha,n} = V_K\big[(I-\Lambda_K^2)^{1/2} \quad 0\big]V_{E}\tr,
    \end{align}
    where in \eqref{equ:B_tilde} and \eqref{equ:B_tilde_svd} we used the same arrangement of singular and eigenvalues. Substituting \eqref{equ:B_tilde_svd} back to \eqref{equ:P_reform} we get the eigendecomposition of $P(\xi)$ as
    \begin{align}
        P(\xi)=&V_{E}\big[(I-\Lambda_K^2)^{1/2} \quad 0\big]\tr\Lambda'_K(\xi)\notag\\
        &\cdot \big[(I-\Lambda_K^2)^{1/2} \quad 0\big]V_{E}\tr,
    \end{align}
    since $V_{E}$ is orthonormal and $[(I-\Lambda_K^2)^{1/2} \quad 0\big]\tr\Lambda'_K(\xi)[(I-\Lambda_K^2)^{1/2} \quad 0\big]$ is diagonal. Recall that $\Lambda'_K(\xi)$ is positive definite and that we assume $\forall\,i: 0<\abs{\lambda_i(K)} < 1$, therefore $P(\xi)$ has exactly $d$ positive eigenvalues, hence $\text{rank}(P(\xi)) = d$. Then, the nonzero eigenvalues of $P(\xi)$ are given as
    \vspace{-1mm}
    \begin{align}
            \lambda_i(P(\xi)) &= \left(\tfrac{\lambda_i^2(K)}{1-\lambda_i^2(K)-\tfrac{1}{\xi}}+1\right)(1-\lambda_i^2(K))\notag\\
            &=\frac{\xi\lambda_i^2(K) - \xi - \lambda_i^2(K) + 1}{\xi\lambda_i^2(K) - \xi + 1},
    \end{align}
    where we used 
    \eqref{equ:def_Lam_K_coma} and that $\Lambda'_K(\xi)$ is invertible.
\end{proof}
 
\bibliographystyle{IEEEtran}
\bibliography{sps}
\end{document}